\documentclass[10pt,a4paper,oneside]{amsart}%
\usepackage[utf8]{inputenc}%
\usepackage[T1]{fontenc}
\usepackage{lmodern}
\usepackage[foot]{amsaddr}%
\usepackage{amssymb,amsmath,amsthm,amstext,amscd,mathtools}%
\usepackage[margin=2.0cm, left=2.5cm,right=2.5cm]{geometry}%
\usepackage{graphicx,subcaption}%

\usepackage{todonotes}%

\usepackage{hyperref}%
\usepackage[capitalise]{cleveref}%
\crefname{figure}{Figure}{Figures}%
\crefname{lemma}{Lemma}{Lemmas}%
\crefname{remark}{Remark}{Remarks}%
\crefname{theorem}{Theorem}{Theorems}%
\crefname{appendix}{Appendix}{Appendices}%
\crefname{equation}{}{}%

\usepackage[normalem]{ulem}%
\usepackage{paralist}%
\theoremstyle{plain}%
\newtheorem{theorem}{Theorem}[section]%
\newtheorem{lemma}[theorem]{Lemma}%
\theoremstyle{definition}%
\theoremstyle{remark}%
\newtheorem{remark}[theorem]{Remark}%
\theoremstyle{example}%
\numberwithin{equation}{section}%
\usepackage[frozencache=true]{minted}%
\definecolor{bg}{rgb}{0.95,0.95,0.95}%

\usepackage[%
backend=bibtex,%
maxbibnames=9,%
style=numeric,%
natbib=true,%
url=true,%
doi=true,%
isbn=false,%
eprint=true,%
firstinits=true]{biblatex}%
\DeclareMathOperator{\E}{\mathbf{E}}%
\DeclareMathOperator{\var}{Var}%
\DeclareMathOperator{\cov}{Cov}%
\newcommand{\dd}{\mathrm{d}}%
\DeclareMathOperator{\heaviside}{Heaviside}%
\newcommand{\rset}{\textbf{R}}%
\newcommand{\given}{\mid}%
\newcommand{\defeq}{\mathrel{\mathop:}=}%
\newcommand{\eqdef}{=\mathrel{\mathop:}}%

\newcommand{\Err}{\mathrm{Err}}%
\newcommand{\bY}{\boldsymbol{Y}\!}%
\newcommand{\bZ}{\boldsymbol{Z}}%
\newcommand{\bz}{\boldsymbol{z}}%
\newcommand{\by}{\boldsymbol{y}}%
\newcommand{\bxi}{\boldsymbol{\xi}}%
\newcommand{\bmu}{\boldsymbol{\mu}}%
\newcommand{\bSigma}{\boldsymbol{\Sigma}}%
\newcommand{\DeltaSYSY}{\mathcal{S}(\bY_{\cdot})^2_{t_i, s}}%
\newcommand{\DeltaSY}{\mathcal{S}(\bY_{\cdot})_{t_i, s}}%
\newcommand{\DeltaSYk}{\mathcal{S}(\bY_{\cdot})^k_{t_i, s}}%
\newcommand{\vertiii}[1]{{\left\vert\kern-0.25ex\left\vert\kern-0.25ex\left\vert
          #1
        \right\vert\kern-0.25ex\right\vert\kern-0.25ex\right\vert}}%
\newcommand{\abs}[1]{\left \lvert #1 \right \rvert}
\newcommand{\norm}[1]{\left \lVert #1 \right \rVert}

\makeatletter
\newsavebox\myboxA
\newsavebox\myboxB
\newlength\mylenA

\newcommand*\xoverline[2][0.75]{%
    \sbox{\myboxA}{$\m@th#2$}%
    \setbox\myboxB\null%
    \ht\myboxB=\ht\myboxA%
    \dp\myboxB=\dp\myboxA%
    \wd\myboxB=#1\wd\myboxA%
    \sbox\myboxB{$\m@th\overline{\copy\myboxB}$}%
    \setlength\mylenA{\the\wd\myboxA}%
    \addtolength\mylenA{-\the\wd\myboxB}%
    \ifdim\wd\myboxB<\wd\myboxA%
       \rlap{\hskip 0.5\mylenA\usebox\myboxB}{\usebox\myboxA}%
    \else
        \hskip -0.5\mylenA\rlap{\usebox\myboxA}{\hskip 0.5\mylenA\usebox\myboxB}%
    \fi}
  \makeatother

\renewcommand{\bar}{\xoverline}%
\title[Weak error rates linear rough volatility]{Weak error rates for option
  pricing under linear rough volatility}%

\author{Christian Bayer$^1$}%
\address{$^{1)}$WIAS, Mohrenstr.~39, 10117 Berlin, Germany.}%
\author{Eric Joseph Hall$^{2}$}%
\address{$^{2)}$University of Dundee, School of Science and Engineering, Mathematics Division, Dundee DD1 4HR, UK.}%
\author{Ra\'ul Tempone$^{3,4}$}%
\address{$^{3)}$RWTH Aachen University, Chair of Mathematics for
  Uncertainty Quantification, Pontdriesch 14-16, 52062 Aachen, Germany.}%
\address{$^{4)}$King Abdullah University of Science and Technology
  (KAUST), Computer, Electrical and Mathematical Sciences \&
  Engineering Division (CEMSE), Thuwal 23955-6900, Saudi Arabia.}%
\email{christian.bayer@wias-berlin.de, ehall001@dundee.ac.uk,
  tempone@uq.rwth-aachen.de}%

\keywords{rough volatility, option pricing, weak error,
  Euler--Maruyama, non-Markovian dynamics, rough Stein--Stein model.}%
\subjclass[2010]{91G60, 91G20, 65C20}%
\begin{document}%

\begin{abstract}
In quantitative finance, modeling the volatility structure of underlying assets is vital to pricing options. Rough stochastic volatility models, such as the rough Bergomi model [Bayer, Friz, Gatheral, Quantitative Finance 16(6), 887-904, 2016], seek to fit observed market data based on the observation that the log-realized variance behaves like a fractional Brownian motion with small Hurst parameter, $H < 1/2$, over reasonable timescales. Both time series of asset prices and option-derived price data indicate that $H$ often takes values close to $0.1$ or less, i.e., rougher than Brownian motion. This change improves the fit to both  option prices and time series of underlying asset prices while maintaining parsimoniousness. However, the non-Markovian nature of the driving fractional Brownian motion in rough volatility models poses severe challenges for theoretical and numerical analyses and for computational practice. While the explicit Euler method is known to converge to the solution of the rough Bergomi and similar models, its strong rate of convergence is only $H$. We prove rate $H + 1/2$ for the weak convergence of the Euler method for the rough Stein--Stein model, which treats the volatility as a linear function of the driving fractional Brownian motion, and, surprisingly, we prove rate one for the case of quadratic payoff functions. Indeed, the problem of weak convergence for rough volatility models is very subtle; we provide examples demonstrating the rate of convergence for payoff functions that are well approximated by second-order polynomials, as weighted by the law of the fractional Brownian motion, may be hard to distinguish from rate one empirically. Our proof uses Talay--Tubaro expansions and an affine Markovian representation of the underlying and is further supported by numerical experiments. These convergence results provide a first step toward deriving weak rates for the rough Bergomi model, which treats the volatility as a nonlinear function of the driving fractional Brownian motion.
\end{abstract}

\maketitle%

\section{Introduction}
\label{sec:introduction}

\emph{Rough} stochastic volatility models form an increasingly popular
paradigm in quantitative finance, as they simultaneously address two
empirical challenges. Firstly, time series of realized variance
indicate that variance is rough in the sense of having Hölder
regularity $H \ll 1/2$, see
\cite{GatheralEtAl:2018vr,BennedsenLundePakkanen:2016sv,FukasawaTakabatakeWestphal:2019aa}.
Secondly, rough volatility models recover the power-law explosion of
the at the money %
implied volatility skew of the form $\tau^{-\gamma}$ for
$\gamma \sim 1/2$ as time to maturity $\tau \to 0$. In fact, these two
constants are linked by $\gamma = 1/2 - H$, giving further evidence of
regularity $H$ being small, say around $0.1$. We refer to
\cite{BayerFrizGatheral:2016rv} for the pricing perspective.

To fix notation, we consider a rough stochastic volatility model for
an asset price process $S_t$ of the form
\begin{equation*}
  \label{eq:roughVol_asset}
  \dd S_t = \sqrt{v_t} S_t \dd Z_t\,,
\end{equation*}
where $Z$ is a Brownian motion (Bm). There are two classes of rough
volatility models which differ in the specification of the
instantaneous variance component $v_t$. The \emph{rough Heston model}
(\cite{ElEuchRosenbaum:2019rh}) is an example of one kind, with $v_t$
given as a solution to a Volterra stochastic differential equation
(SDE) with a power law kernel $K(r) \sim r^{H-1/2}$, $r>0$. This paper
will consider an alternative where the variance process is an explicit
function of a fractional Brownian motion (fBm) $W^H_t$, which does not
need to be the classical fBm. For instance, the \emph{rough Bergomi
  model} (\cite{BayerFrizGatheral:2016rv}) is specified by the choice
\begin{equation}
  \label{eq:roughBergomi_v}
  v_t \defeq \xi(t) \exp\left( \eta W^H_t - \frac{1}{2} \eta^2 t^{2H} \right),
\end{equation}
where $\xi(t)$ denotes the \emph{forward variance} and $W^H_t$ denotes
the \emph{Riemann--Liouville} fBm given by
\begin{equation}
  \label{eq:fBm}
  W^H_t \defeq \int_0^t K(t-s) \dd W_s,
  \quad K(r) \defeq \sqrt{2H} r^{H-1/2},
\end{equation}
for a Bm $W$ with correlation $\rho$ with $Z$. A related
  model where the variance process is an explicit function of the fBm
  is the \emph{fractional} or \emph{rough Stein--Stein model}
  (\cite{AbiJaber:2021ss}), given by
\begin{align}
  &\dd S_t = v_t S_t \dd Z_t\,, \notag \\
  &v_t = v_0(t) + \int_{0}^{T} K(t,s) \kappa v_s \dd{s}
    + \int_{0}^{T} K(t,s) \eta \dd W_s\,, \label{eq:fSS-vol}
\end{align}
for a Volterra kernel $K$ and for
arbitrary correlation $\rho$ between $Z$ and $W$. This extends the classic Stein--Stein model (\cite{SteinStein:1991fm}) and its generalization (\cite{SchoebelZhu:1999sv}). For the particular choice $\kappa = 0$
and $K(t,s) = K(t-s)$ from \cref{eq:fBm}, the volatility term
\cref{eq:fSS-vol} is a linear function of the fBm,
\begin{equation*}
  \label{eq:fSS-vol-reduc}
  v_t = v_0(t) + \eta W_t^H\,.
\end{equation*}
For the later volatility, the rough Stein--Stein model can be
viewed as a simplified fractional SABR model that enables explicit
computations of certain quantities of interest (\cite{gatheral-bbq-slides}).

The modelling advantages gained by capturing these two empirical
challenges, i.e., low H\"older regularity ($H \ll 1/2$) and the
power-law explosion, using a rough stochastic volatility model are
paid for both on the theoretical and the numerical side. Indeed, rough
stochastic volatility models are neither semi-martingales nor Markov
processes. Despite the former, rough volatility models do not violate
the no-arbitrage-condition, as the asset price process itself is a
martingale. On the other hand, the difficulties caused by the lack of
Markov property are more severe. In particular, there is no finite
dimensional pricing PDE anymore (although we refer to
\cite{JacquierOumgari:2019sv,bayer2020pricing} for implementations of
an infinite-dimensional pricing PDE based on machine learning). For
some rough volatility models of affine Volterra type, for instance,
the rough Heston model, there is still a semi-explicit formula for the
asset price's characteristic function in terms of a deterministic
fractional ODE. Otherwise, the rough stochastic volatility approach
necessitates simulation-based methods.

On the numerical side, $W^H_{t_1}, \ldots, W^H_{t_N}$ can be exactly
sampled at discrete-time points as $W^H$ is a Gaussian process with
known covariance function. (The hybrid scheme of
\cite{BennedsenLundePakkanen:2017hs} is a popular alternative to exact
simulation, sacrificing accuracy for speed.) However, simulation of
$S_t$ requires discretization of a stochastic integral,
even in the case of the rough Bergomi and rough Stein--Stein models. 
As we shall see in further detail
later, we essentially need to compute stochastic integrals of the form
\begin{equation}
  \label{eq:stoch_int_intro} \int_0^T \psi(t, W^H_t) \dd W_t,
\end{equation}
for some deterministic, `nice' function $\psi$. In particular, note
that the integrand is adapted and square-integrable (under appropriate
conditions). Hence, the stochastic integral exists in the classical
It\={o} sense, and strong convergence of the numerical scheme
\begin{equation}
  \label{eq:stoch_int_intro_approx}
  \sum_{i=0}^{n-1} \psi(t_i, W^H_{t_i}) (W_{t_{i+1}} -W_{t_i})
\end{equation}
is also classical. The speed of convergence is considerably less
clear. Indeed, Neuenkirch and Shalaiko~\cite{Neuenkirch:2016ob} proved
strong convergence with rate $H$ for a very similar problem,
i.e.,~phrased in terms of classical fBm, and strong rate $H$ is widely
expected to hold also for the approximation scheme
\cref{eq:stoch_int_intro_approx} to \cref{eq:stoch_int_intro}. Using
techniques from regularity structures, in particular, renormalization
by an exploding constant, \cite{BayerEtAl:2020rs} proved essentially
the same strong rate for a Wong--Zakai type approximation of
\cref{eq:stoch_int_intro}.

Combining our observations---that volatility is rough
($H \approx 0.1$) and typical schemes converge with strong rate $H$---
we run into problems, as the rate of convergence is so small as to
make it indistinguishable from lack of convergence in many cases of
practical importance. Indeed, suppose that $H = 0.1$ and we need $n$
time steps to reach an error tolerance $\epsilon$. If we now decrease
our tolerance by a factor ten, i.e., we require one additional
significant digit, then the number of time-steps needed is increased
by a factor $10^{10}$ in the asymptotic regime.

For most applications we really require weak as opposed to strong
convergence of the numerical scheme. For instance, the price of a
European option with payoff $\varphi$ is $\E[\varphi(S_T)]$, and its
computation relies on weak convergence of the scheme. %
Weak approximation of stochastic integrals is often much faster than
strong approximation. Consider the Euler scheme for standard SDEs (the
case $H = 1/2$). Generically, i.e., when the problem is sufficiently
`nice', the weak rate of convergence is one, whereas the strong rate
is $1/2$. This poses the interesting question about the relation
between the Hölder regularity ($H = 1/2$), the weak rate of
convergence ($1$) and the strong rate of convergence ($1/2$). Indeed,
\cite{Neuenkirch:2016ob} showed us that the strong rate is equal to
the Hölder regularity $H$, but there are several plausible candidates
for the weak rate: $2H$, $H+1/2$, and $1$ (independent of
$H$).\footnote{Anecdotally, we asked several experts on stochastic
  numerics in early stages of working on this problem, and all three
  possibilities were put forward.} We stress that only the last two
alternatives allow for feasible numerical simulations in the truly
rough regime. Bluntly put, if the true weak error only decays
proportionally to $n^{-2H}$ in the number of time steps $n$, then
simulation methods are not viable numerical methods for option pricing
in rough volatility models.

Despite the importance of the problem of determining the weak rate,
only little work has been done. Horvath, Jacquier and Muguruza
\cite{HorvathJacquierMuguruza:2019aa} study a Donsker theorem for a
rough volatility model, which translates into a week tree-type
approximation. The rate of convergence of their method is $H$ in the
number of time-steps. At this stage, we should note that the trees are
non-recombining, implying that the memory load increases exponentially
in the number of time-steps. To the best of our knowledge, this work
provides the only rigorous weak convergence result in the literature
of rough volatility models. Indeed, it is worth pointing out that
standard proof techniques for diffusions, see \cite{TalayTubaro:1990},
strongly rely on the Markov property, and are, hence, not applicable
in this setting.

At the same time, discretization-based simulation methods are often
used in the literature, with great success. While convergence is
rarely considered (not even empirically), we would expect to see
difficulties emerge in the very rough cases $H \approx 0.1$ if the
convergence rates were truly as bad as only $H$ or $2H$. In fact, the
few available empirical studies (for instance,
\cite{BayerEtAl:2020rb}) indicate a much larger weak rate of
convergence. In fact, the authors of \cite{BayerEtAl:2020rb} observe a
weak rate of one which is stable enough to allow accelerated
convergence by Richardson extrapolation.

In this paper we prove novel weak rates for the convergence of the
left-hand rule \cref{eq:stoch_int_intro_approx} to
\cref{eq:stoch_int_intro}:
\begin{theorem}
  \label{thr:main}
  The left-point approximation \cref{eq:stoch_int_intro_approx} to the
  rough stochastic integral \cref{eq:stoch_int_intro} converges with
  weak rate $H+1/2$ for $\psi(t,W^H_t) = W^H_t$ -- i.e., in the rough Stein--Stein model. For the case that the
  payoff $\varphi$ is a quadratic polynomial the convergence is with
  weak rate one.
\end{theorem}
We refer to \cref{thm:weak-rate-gen,thm:weak-rate-quad} for more
precise statements. Some remarks are in order:
\begin{itemize}
\item The problem of weak convergence in this setting is very subtle;
  if we restrict ourselves to quadratic polynomials as payoff
  functions, then the weak rate of convergence is actually one, see
  \cref{lem:weak-rate-quad-simple}. This implies the rate of
  convergence for payoff functions $\varphi$ that can be well
  approximated by quadratic polynomials, as seen from the law of the
  solution, may be hard to distinguish from rate one empirically, due
  to prevalence of higher order terms (see \cref{fig:shifted-cubic}).
  Note that the result \ref{lem:weak-rate-quad-simple} and its proof
  were communicated to us by Andreas Neuenkirch \cite{neuenkirch-comm}
  prior to starting this work; \cref{lem:weak-rate-quad-simple}
  indicates rate one for quadratic payoffs $\varphi$ for a more general class
  of $\psi$ (i.e., including rough Bergomi) but it is
  unclear how to generalize this result to a broader class of payoffs.
\item We do not have a lower bound establishing that the weak rate of
  convergence cannot be better than $H+1/2$ in the generic case. We do
  offer numerical evidence for this assertion, though, see
  \cref{fig:conv-general-payoff,fig:conv-rate-one,fig:shifted-cubic}.
\item We do not doubt that the proof extends to the general case of
  non-linear $\psi$, which includes the rough Bergomi model. Indeed,
  the present paper is partly motivated to expose a possible proof
  strategy for the general case. Extending the method of proof using
  Fa\`{a} di Bruno's formula poses some technical challenges, mainly
  due to the needed to control more complicated formulas.
\end{itemize}

Our proof for \cref{thr:main} relies on deriving Taylor expansions for
the weak error using an affine Markovian representation of the
underlying. The basic flavor of this approach, i.e., obtaining a
Markovian extended variable system to facilitate analysis, is a
strategy utilized in other non-Markovian stochastic dynamical systems
such as the Generalized Langevin equation (see,
e.g.,~\cite{Halletal:2016uq,DidierEtAl:2012sc}) and open Hamiltonian
systems (\cite{Rey-Bellet:2006oc}). In the context of rough volatility
models, Markovian approximations were also used
in~\cite{abi2019multifactor}.

\subsection*{Outline of the paper}

In \cref{sec:main-result} we provide the setting and the main result
and discuss the general strategy of the proof.
\Cref{sec:extended-state-space} introduces auxiliary, Markovian
approximations to both~\eqref{eq:stoch_int_intro} and
\cref{eq:stoch_int_intro_approx} based on \cite{CarmonaCoutin:1998bm}.
This high dimensional Markovian problem will serve as a surrogate
problem for most of the convergence analysis.
\Cref{sec:weak-rate-quad} considers the special case of quadratic
payoff functions, for which the general proof strategy simplifies
considerably. We contrast this with a specific proof only applicable
to quadratic payoffs, which also works for general non-linear $\psi$.
The proof of \cref{thr:main} (and \cref{thm:weak-rate-gen}) is then
carried out in \cref{sec:weak-rate-gen}.
  
\section{Problem setting: weak rate of convergence for Euler scheme
  is $H+1/2$}%
\label{sec:main-result}%

We consider a smooth, bounded payoff function $\varphi(X_T)$ for an
underlying
\begin{equation}
  \label{eq:underlying}
  X_t \defeq \int_0^t \psi(s, W^H_s) \dd W_s\,,
\end{equation}
where $W^H_t$ is a Riemann--Liouville fBm given by \cref{eq:fBm} with
Hurst parameter $H \in (0, 1/2)$.
A simplified model of rough stochastic volatility,
\cref{eq:underlying} retains keys features of the rough Bergomi model \cref{eq:roughBergomi_v} and the rough Stein-Stein model \cref{eq:fSS-vol}. Namely, the $X_t$ in
\cref{eq:underlying} is non-Markovian as $W^H_t$, and hence
$\psi(t, W^H_t)$, depends on the full history of $(W_s)_{s \in [0,t]}$
(cf.\ $\psi$ to the instantaneous variance $v_t$ in
\cref{eq:roughBergomi_v}). In fact, for the purposes of European
option, the rough Bergomi model can be reduced to~\cref{eq:underlying}
in the following way (often attributed to \cite{RomanoTouzi:1997sv}).
First, It\={o}'s formula implies that
\begin{equation*}
  \label{eq:RT-derivation-1}
  S_T = S_0 \exp\left( - \frac{1}{2} \int_0^T v_s \dd s + \int_0^T \sqrt{v_s}
    \dd Z_s \right).
\end{equation*}
We can now replace the Bm $Z$ by $\rho W + \sqrt{1-\rho^2} W^{\perp}$ for an
independent Bm $W^{\perp}$. \emph{Conditionally on $W$}, $S_T$ has a
log-normal distribution with parameters
\begin{equation*}
  \label{eq:RT-derivation-2}
  \mu \defeq \log S_0 - \frac{1}{2} \int_0^T v_s \dd s + \rho \int_0^T \sqrt{v_s}
  \dd W_s, \quad \sigma^2 \defeq  (1 - \rho^2) \int_0^T v_s \dd s.
\end{equation*}
If we denote the Black--Scholes price for the payoff function
$\varphi$ at maturity $T$ by $C_{BS}(S_0, \sigma_{BS}^2T, \varphi)$,
for interest rate $r = 0$ and volatility $\sigma_{BS}$, then we get
\begin{equation}
  \label{eq:RT-derivation-3}
  \E[ \varphi(S_T) ] = \E\left[ C_{BS}\left( S_0 \exp\left[  -
        \frac{\rho^2}{2} \int_0^T v_s \dd s + \rho \int_0^T \sqrt{v_s} \dd W_s
      \right], \, (1 - \rho^2) \int_0^T v_s \dd s, \, \varphi \right)
  \right].
\end{equation}
Computation of the right hand side of~\cref{eq:RT-derivation-3}
requires simulation of the Lebesgue integral $\int_0^T v_s \dd s$ as
well as simulation of
\begin{equation}
  \label{eq:RT-derivation-4}
  \int_0^T \sqrt{v_s} \dd W_s = \int_0^T \sqrt{\xi(s)} \exp\left(
    \frac{\eta}{2} W^H_s - \frac{\eta^2}{4} s^{2H} \right) \dd W_s,
\end{equation}
which is of the form~\cref{eq:underlying}.

Presently, we derive weak rates of convergence,
\begin{equation}
  \label{eq:weak-rate-formal}
  \bigl| \E[\varphi(X_T) - \varphi(\bar{X}_{T}^{\Delta t})] \bigr|
  = O(\Delta t^\gamma)\,,
\end{equation}
for the left-hand scheme \cref{eq:stoch_int_intro_approx} with
step-size $\Delta t$ such that $n \Delta t = T$. Restricting to the
the rough Stein--Stein model $\psi(s, W^H_s) = W^H_s$,
the main finding of this work, in \cref{thm:weak-rate-gen} (and
implying the first statement in \cref{thr:main}), is that the weak
rate is $\gamma = H + 1/2$ for the Hurst parameter $H$.

\begin{theorem}[Weak rate]
  \label{thm:weak-rate-gen}
  For general $\varphi \in C_b^\eta$, for integer
    $\eta = \lceil \tfrac{1}{H} \rceil$, and the rough Stein--Stein model $\psi(s, W^H_s) = W^H_s$, we have
  \begin{equation*}
    \label{eq:weak-rate-gen}
    \left|\Err(T, \Delta t) \right|
    = \bigl| \E [\varphi(X_T) - \varphi(\bar{X}^{\Delta t}_{T})] \bigr|
    = O(\Delta t^{H+1/2})\,,
  \end{equation*}
  i.e.~the Euler method is weak rate $H+1/2$.
\end{theorem}

The proof of \cref{thm:weak-rate-gen} is presented in
\cref{sec:weak-rate-gen}. Before diving into the machinery needed for
the proof, we first consider some numerical evidence that supports the
rates in \cref{thr:main} and the accompanying remarks. Details of the
implementation are outlined in \cref{sec:numerics}.

The first group of numerical experiments, in
\cref{fig:conv-general-payoff}, provide support for rate $H+1/2$ in
\cref{thm:weak-rate-gen}. In \cref{fig:conv-general-payoff}, the weak
error rate is observed to depend on $H$ for the general
(i.e.~non-quadratic) payoff functions $\varphi(x) = x^3$ and
$\varphi(x) = \heaviside(x)$. Indeed, the best fits (least squares) of
the weak error to $\Delta t$, as well as the extremes suggested by the
upper and lower 95\% confidence interval for the mean based on
$M=3\times 10^6$ samples, is consistent with the rate $H+1/2$.
Comparing \cref{fig:conv-general-payoff:H.05} to
\cref{fig:conv-general-payoff:H.15}, the rate increases (and by
approximately $0.1$) as $H$ increases from $H=0.05$ to $0.15$.
Although the function $\varphi(x) = \heaviside(x)$ is not continuous
and therefore does not fit precisely into our theory, the consistency
of the observed rates in \cref{fig:conv-general-payoff} hint at the
generality of the findings in \cref{thm:weak-rate-gen} to, e.g.,
digital call options.

\begin{figure}[bhtp]
  \centering
  \begin{subfigure}[t]{0.49\textwidth} \includegraphics[width=1\textwidth]{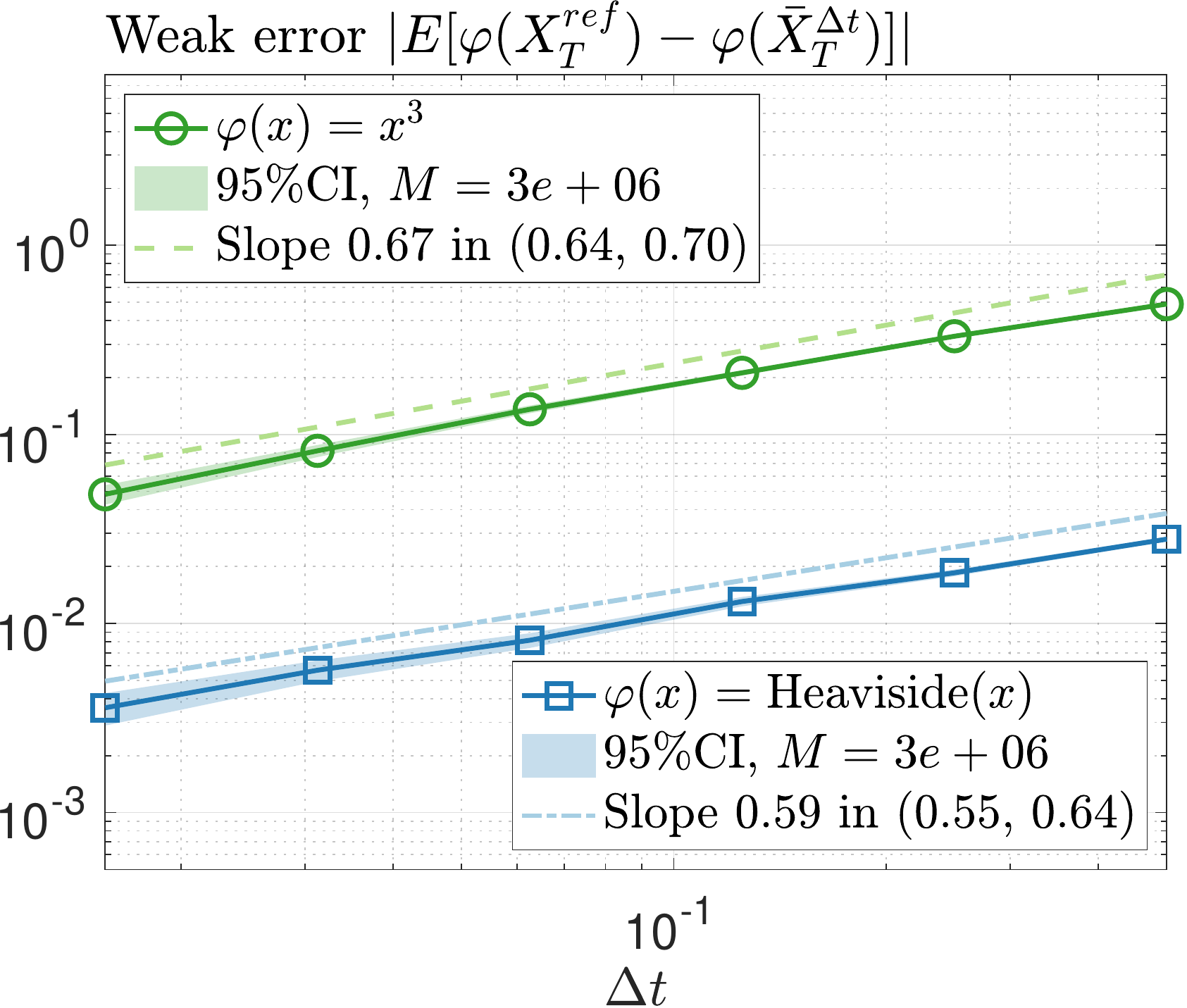}
\caption{$H=0.05$}
   \label{fig:conv-general-payoff:H.05}
    \end{subfigure}%
    \hfill
    \begin{subfigure}[t]{0.49\textwidth}
      \includegraphics[width=1\textwidth]{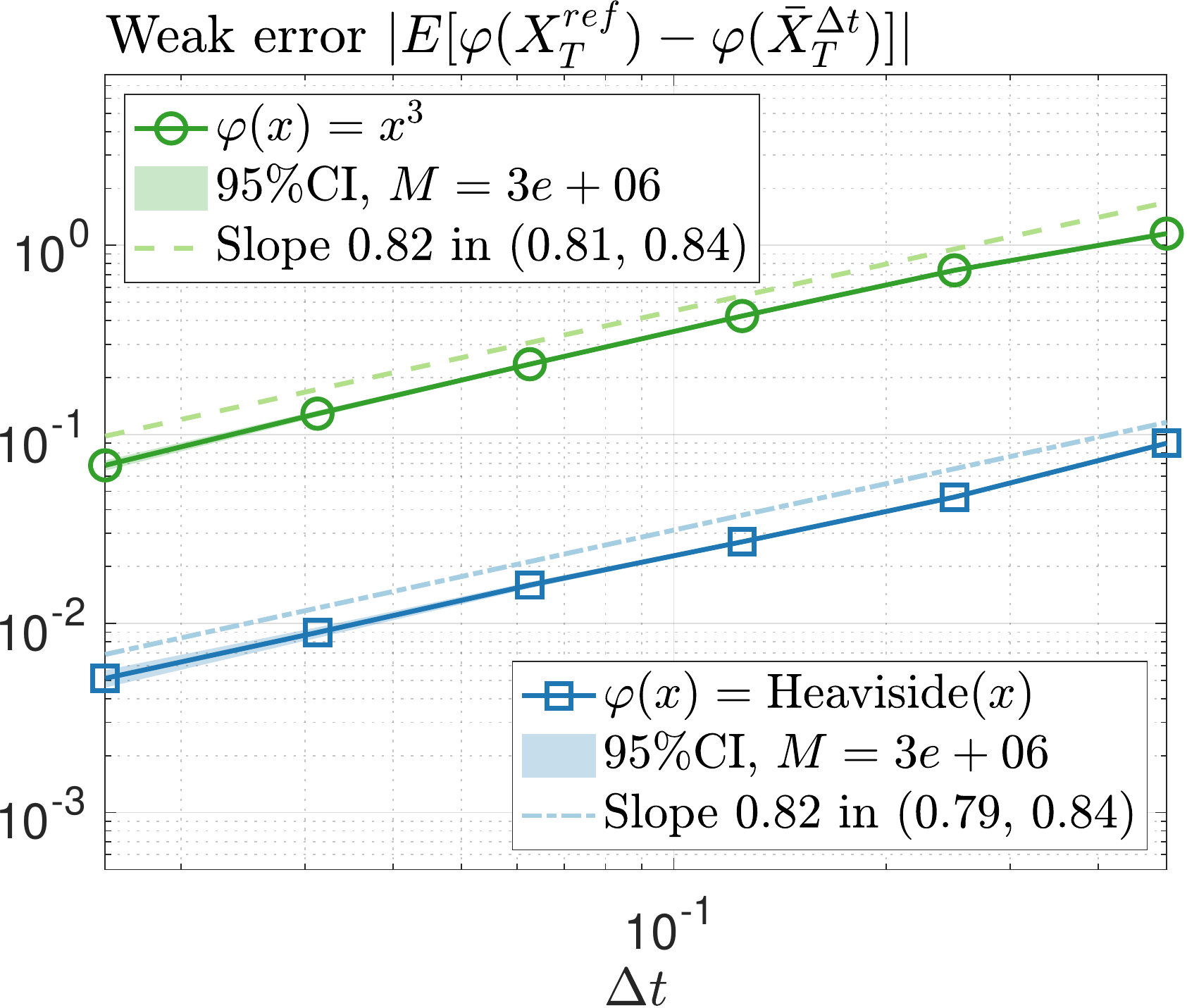}
        \caption{$H=0.15$}
  \label{fig:conv-general-payoff:H.15}
      \end{subfigure}
      \caption{For small Hurst parameters,
        (\subref{fig:conv-general-payoff:H.05}) $H=0.05$ and
        (\subref{fig:conv-general-payoff:H.15}) $H=0.15$, the best fit
        slope for the weak error for scheme
        \cref{eq:stoch_int_intro_approx}, together with extremes
        suggested by the $95\%$ CI based on $M$ observations, are
        consistent with the rate $H+1/2$ obtained in
        \cref{thm:weak-rate-gen} for general payoff functions
        $\varphi$. In particular, the rate holds for the discontinuous
        $\varphi(x)=\heaviside(x)$ suggesting our findings are robust.
        Here $\Delta t \in [2^{-6}, 2^{-1}]$ and the reference mesh is
        $\Delta t^{ref} = 2^{-12}$.}
  \label{fig:conv-general-payoff}
\end{figure}

In \cref{fig:conv-Bm}, we observe that for $H=1/2$, i.e.~standard
Brownian motion, the best fit of the weak error rate is consistent
with the known weak rate one for general payoff functions. However, in
contrast to the rates observed in \cref{fig:conv-general-payoff}, the
behavior of quadratic payoffs looks decidedly different. We observe in
\cref{fig:conv-quadratic-payoff} that the weak rate for quadratic
$\varphi(x) = x^2$ appears to be $\gamma=1$ even for small $H=0.05$
and $H=0.15$. Weak rate one for quadratic payoff functions is recorded
in \cref{thm:weak-rate-quad} and \cref{lem:weak-rate-quad-simple} in
\cref{sec:weak-rate-quad}; this surprising finding, that the rate
depends on the payoff function, will be readily explained using the
asymptotic expansions that are at the center of our approach.

Finally, in \cref{fig:shifted-cubic}, we observe that the best fit of
weak rate to $\Delta t$ for the shifted-cubic $\varphi(x) = (x+1.5)^3$
is consistent with rate $1$ even for small $H=0.05$ and $H=0.15$
(cf.~compare the rates in \cref{fig:shifted-cubic} to those for the
cubic payoff $\varphi(x) = x^3$ in
\cref{fig:conv-general-payoff:H.05,fig:conv-general-payoff:H.15}). As
seen from the law of the solution, the shifted cubic is better
approximated by quadratic polynomials and therefore its rate of
convergence is much harder to distinguish from rate one. This
numerical experiment not only drives home the subtlety of the problem
of deriving weak rates for rough stochastic volatility models, but
also leads us to be optimistic that efficient numerical methods can be
obtained for a wide array of real-world problems where the
\emph{effective} rate of convergence is not as bad as the theoretical
rate.

\begin{figure}[bhtp]
  \centering
    \begin{subfigure}[t]{0.49\textwidth}
      \includegraphics[width=1\textwidth]{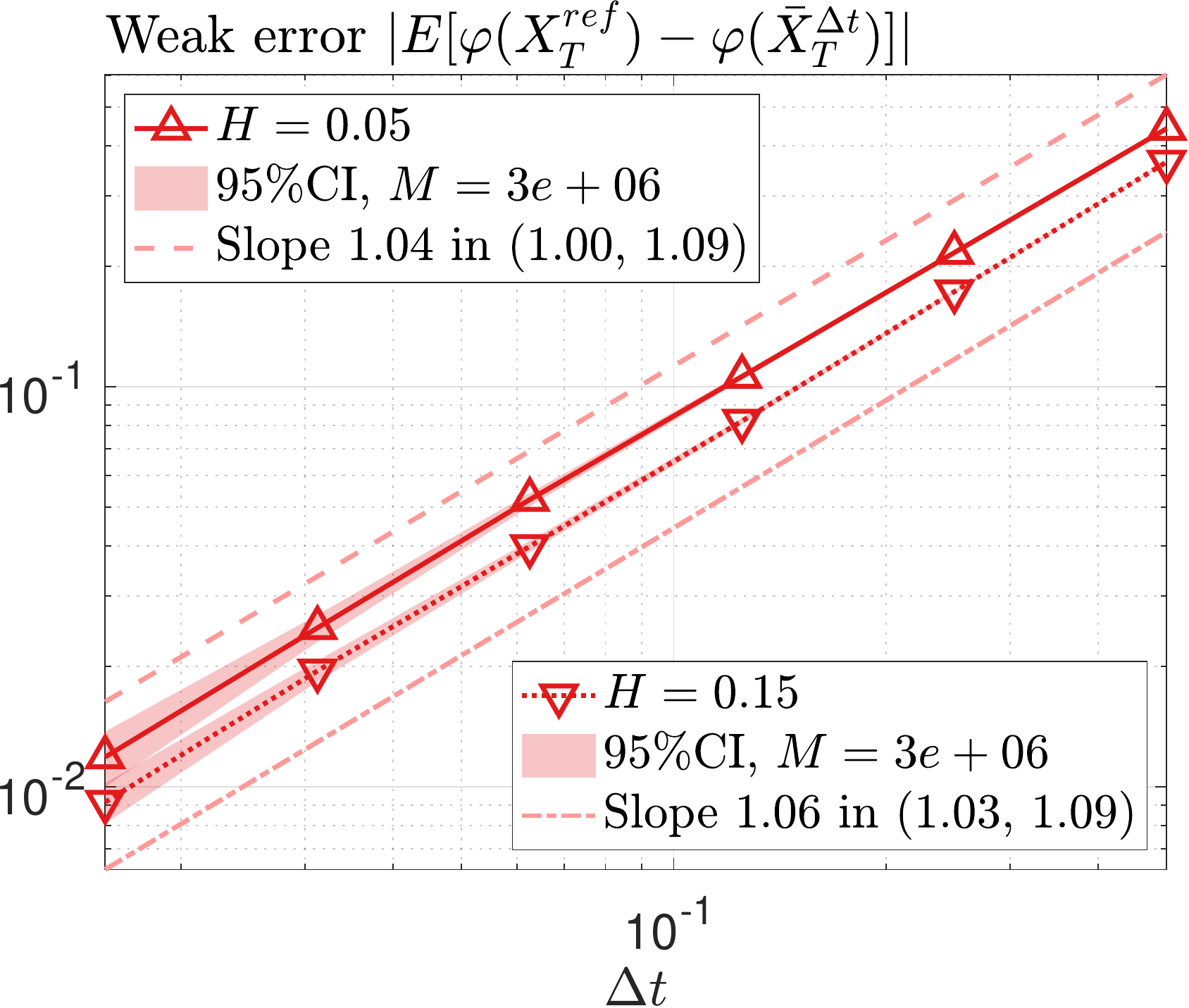}
      \caption{Quadratic $\varphi$, small $H$}
  \label{fig:conv-quadratic-payoff}
\end{subfigure}
\hfill
 \begin{subfigure}[t]{0.49\textwidth}
   \includegraphics[width=1\textwidth]{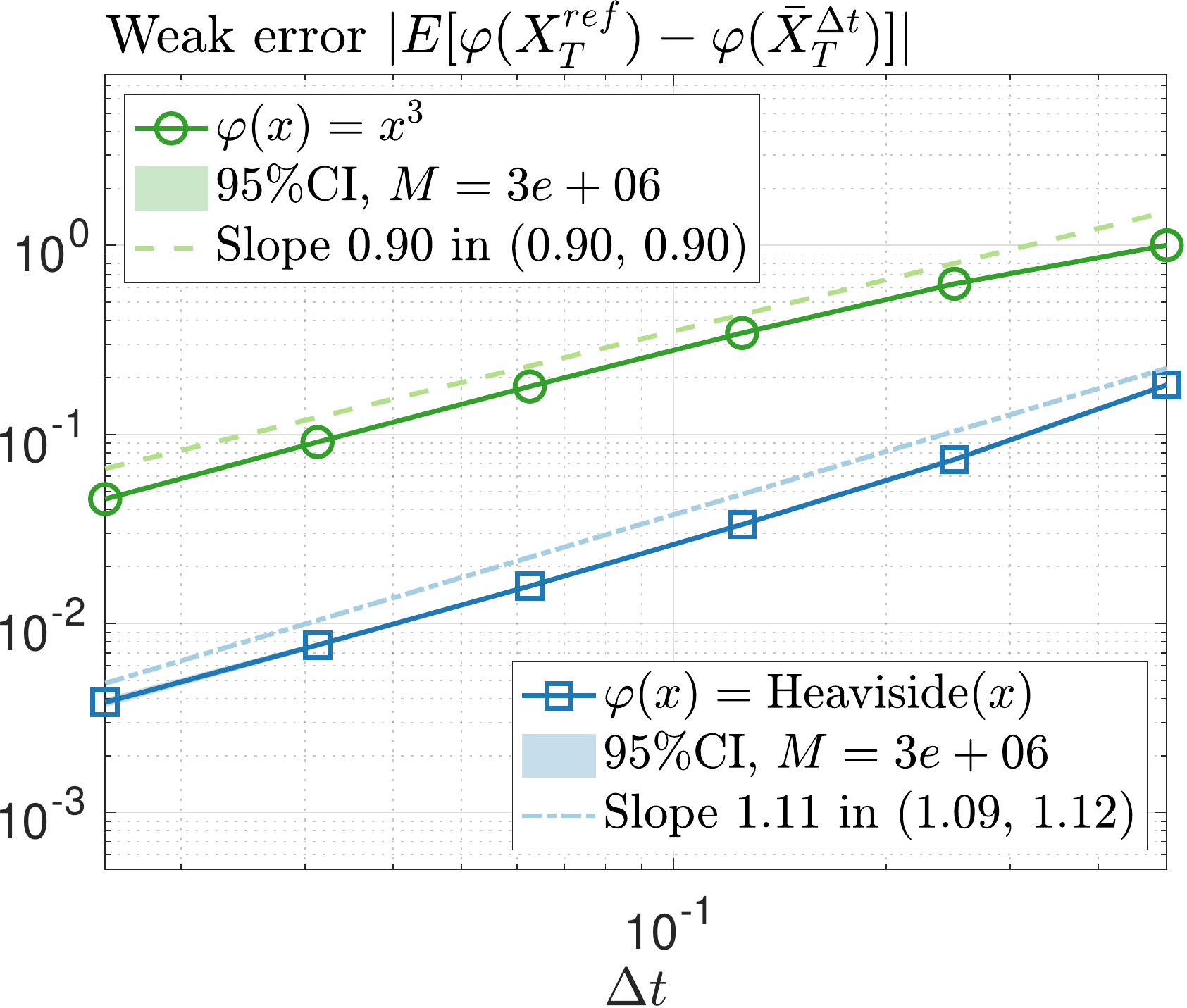} 
   \caption{General (non-quadratic) $\varphi$, $H=0.5$}
   \label{fig:conv-Bm} 
    \end{subfigure}%
    \caption{(\subref{fig:conv-quadratic-payoff}) Surprisingly, the
      best fit line for the weak error for scheme
      \cref{eq:stoch_int_intro_approx} for the quadratic payoff
      $\varphi(x)=x^2$ is consistent with weak rate one even for small
      $H$, as found in \cref{thm:weak-rate-quad}.
      (\subref{fig:conv-Bm}) For Hurst parameter $H=1/2$, the weak
      rate in \cref{thm:weak-rate-gen} for scheme
      \cref{eq:stoch_int_intro_approx} is consistent with the expected
      rate one (for standard Bm), as illustrated by the best fit slope
      for the weak error for $\varphi(x) = x^3$ and
      $\varphi(x) = \heaviside(x)$ (cf.~weak rate $H+1/2$ observed in
      \cref{fig:conv-general-payoff} for small $H$). Here $\Delta t \in [2^{-6}, 2^{-1}]$ and the reference mesh is
        $\Delta t^{ref} = 2^{-12}$.}
      \label{fig:conv-rate-one}
\end{figure}

\begin{figure}[bhtp]
  \centering   \includegraphics[width=0.49\textwidth]{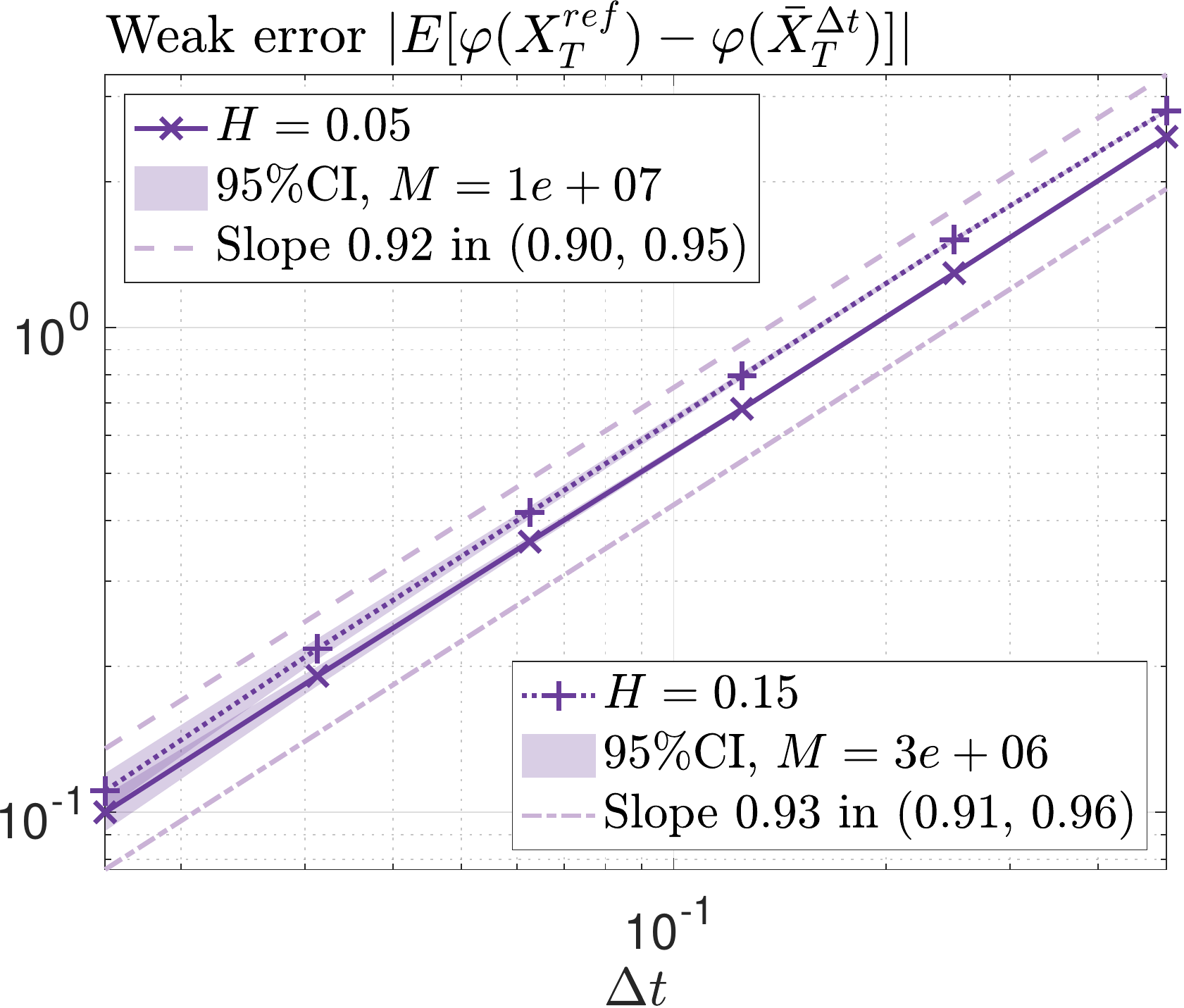}
  \caption{The weak error for scheme \cref{eq:stoch_int_intro_approx}
    for the shifted cubic payoff $\varphi(x) = (x+1.5)^3$ achieves a
    higher rate than $\varphi(x) = x^3$ as the shifted cubic is better
    approximated by a quadratic in the support of the distribution for
    the underlying (cf.~ \cref{fig:conv-general-payoff}). Here
    $\Delta t \in [2^{-6}, 2^{-1}]$ and the reference mesh is
    $\Delta t^{ref} = 2^{-12}$. }
  \label{fig:shifted-cubic}
\end{figure}

\begin{remark}[Financial applications]
  \label{rem:assumptions-main-theorem}
  Although the assumptions of \cref{thm:weak-rate-gen} seem extremely
  strong, they do reflect meaningful financial situations. In
  particular, note that~\cref{eq:RT-derivation-3} allows us to replace
  the (generally non-smooth) payoff functions of European options by
  their smooth Black--Scholes prices. Additionally, put-call-parity
  may allow us to assume bounded payoffs. Linearity of $\psi$ is,
  admittedly, a very strong assumption, which should be seen as the
  first stepping stone to the general result. We conjecture that
  \cref{thm:weak-rate-gen} holds in the setting of the rough Bergomi
  model, i.e., for non-linear $\psi$ as given
  in~\cref{eq:RT-derivation-4}.
\end{remark}

\begin{remark}[Scheme]
  For the simple model problem \cref{eq:stoch_int_intro} the numerical
  integration scheme \cref{eq:stoch_int_intro_approx} is the
  left-point approximation. If the problem were not trivialized to a
  stochastic integral, then in general $\bar{X}^{\Delta t}_{T}$ would
  correspond to the Euler--Maruyama approximation for the underlying
  SDE and we will refer to the scheme interchangeably as both.
\end{remark}

In the next section, we introduce the notation and concepts that will
be used to derive asymptotic expansions for the weak error in powers
of $\Delta t$. In particular, we first use these expansions to derive
weak rate one for quadratic payoffs, see \cref{thm:weak-rate-quad}, in
\cref{sec:weak-rate-quad}. Finally in \cref{sec:weak-rate-gen}, a
proof, following the approach used for \cref{thm:weak-rate-quad} as a
guide, is given for \cref{thm:weak-rate-gen} obtaining weak rate
$H+1/2$ for general payoff functions. Taken together, the statements
of \cref{thm:weak-rate-gen,thm:weak-rate-quad} imply \cref{thr:main}.

\section{Markovian extended state space formulation}
\label{sec:extended-state-space}

We first consider a well-known affine representation for the driving
fBm. Discretizing this affine representation yields an extended state
space for the dynamics of the underlying. A novelty of our approach is
to utilize this formulation to obtain asymptotic expansions for the
weak error. In particular, we utilize the Markovian structure of the
extended state space to show that \cref{eq:weak-error-euler} admits a
Taylor expansion in $\Delta t$ where the coefficients can be
controlled independently of the choice of parameters used to obtain
the extended state space formulation.

\subsection{Affine representations for small Hurst index}
\label{sec:affine-rep-small-H}

Over the Hurst parameter regime of interest, the fBm \cref{eq:fBm}
admits an affine representation as a linear functional of an
infinite-dimensional family of Ornstein--Uhlenbeck (OU) processes
(\cite{CarmonaCoutin:1998bm}).

\begin{lemma}[Affine representation]
  \label{lem:affine-rep}
  For $0 < H < 1/2$,
  \begin{equation}
    \label{eq:affine-rep}
    W^H_t = \widetilde{c}_H
    \int_0^\infty \widetilde{Y}_t(\theta) \theta^{-(H+\frac{1}{2})} \dd \theta \,,
  \end{equation}
  where
  \begin{equation*}
    \widetilde{Y}_{t}(\theta) = \int_{0}^{t} e^{-\theta (t-s)} \dd{W_s} 
  \end{equation*}
  and $\widetilde{c}_H$ is a positive and finite constant depending on $H$.
\end{lemma}
 
Although this statement is well-known we provide key details of the
proof that will be referenced later for the convenience of the reader.
The full proof can be found in,
e.g.,~\cite{CarmonaCoutin:1998bm,HarmsStefanovits:2019fp} (see also
\cite{CarmonaCoutinMontseny:2000gp,Muravlev:2011bm,Harms:2019rb} where
\cite{CarmonaCoutinMontseny:2000gp} gives a Markovian representation
for $H>1/2$, \cite{Muravlev:2011bm} a time-homogeneous Markovian
representation that is also defined for $t \in (-\infty, 0)$, and
\cite{Harms:2019rb} gives bounds on tails and derivatives of the
affine representation).

\begin{proof}
  Writing the kernel appearing in
  \cref{eq:fBm} as a Laplace transform,%
  \begin{equation*}
    \label{eq:kernel-explicit}
    (t-s)^{H-\frac{1}{2}} = \frac{1}{\Gamma(\frac{1}{2} - H)} \int_0^\infty 
    \theta^{-(H+\frac{1}{2})} e^{-\theta(t-s)} \dd \theta\,,
  \end{equation*}
  and then using stochastic Fubini one obtains the desired result,
  \begin{align*}
    W_t^H
    &= \int_0^t \frac{\sqrt{2H}}{\Gamma(\frac{1}{2} - H)}
      \int_0^\infty  \theta^{-(H+\frac{1}{2})} e^{-\theta(t-s)} \dd \theta \dd W_s \\
    &= \int_0^\infty \widetilde{c}_H \int_0^t e^{-\theta (t-s)} \dd W_s
      \theta^{-(H+\frac{1}{2})} \dd \theta \\
    &= \widetilde{c}_H \int_0^\infty \widetilde{Y}_t(\theta) \theta^{-(H+\frac{1}{2})} \dd \theta\,,
  \end{align*}
  where
  $\widetilde{c}_H \defeq \sqrt{2H}/ \Gamma(\frac{1}{2}-H) < \infty$. %
\end{proof}

A key tool in our proof of the weak rates will be to utilize the
Markovian structure of a projection of the fBm obtained by
discretizing the affine representation \cref{lem:affine-rep}. We
observe that the integral in \cref{eq:affine-rep} has a singularity at
$\theta=0$, but behaves essentially like $\theta^{-(H+\frac{1}{2})}$
before $\widetilde{Y}_t(\theta)$ vanishes in the limit of $\theta$.
To make \cref{eq:affine-rep} more amenable to quadrature we remove the
singularity by introducing the change of variable,
\begin{equation*}
  \label{eq:change-vars}
  \vartheta = \theta^{-(H+\frac{1}{2}-1)} = \theta^{\frac{1}{2} - H}\,,
\end{equation*}
thereby obtaining the representation
\begin{equation}
  \label{eq:affine-rep-singularity-removed}
  W_t^H
  = c_H \int_0^\infty \widetilde{Y}_t(\vartheta^{2/(1-2H)}) \dd \vartheta
  = c_H \int_0^\infty Y_t(\theta) \dd \theta \,, 
\end{equation}
where the constant,
\begin{equation*}
  c_H \defeq \frac{\widetilde{c}_H}{\tfrac{1}{2} - H}
  = \frac{\sqrt{2H}}{\Gamma(\tfrac{3}{2} - H)}\,,
\end{equation*}
is an increasing function of $H \in (0,\frac{1}{2})$ such that
$0 < c_H < 1$. In \cref{eq:affine-rep-singularity-removed},
\begin{equation}
  \label{eq:ou-extended-var}
  Y_t(\theta) = \int_{0}^{t} e^{-(t-s)\theta^{p}} \dd{W_s} 
\end{equation}
is an OU process with speed of mean-reversion given by $\theta^{p}$
with a positive power
\begin{equation}
  \label{eq:rev-pow}
  p \defeq 2/(1-2H) > 2 \,.
\end{equation}
One realization of $Y_t(\theta)$ is plotted in \cref{fig:Y-t-theta}
together with an envelope illustrating plus/minus two standard
deviations of $Y_t(\theta)$, computed using the formula for the
covariance, i.e.~
\begin{equation*}
  \label{eq:ou-covar}
  \cov(Y_t(\theta), Y_t(\eta))
  = \frac{1}{\theta^p + \eta^p}(1-e^{-(\theta^p+\eta^p)t})\,.
\end{equation*}
Replacing the integral in \cref{eq:affine-rep-singularity-removed}
with a quadrature rule in the parameter $\theta$ yields a projection
of the fBm onto a finite state space.

\begin{figure}[h]
  \centering  
  \includegraphics[width=\textwidth]{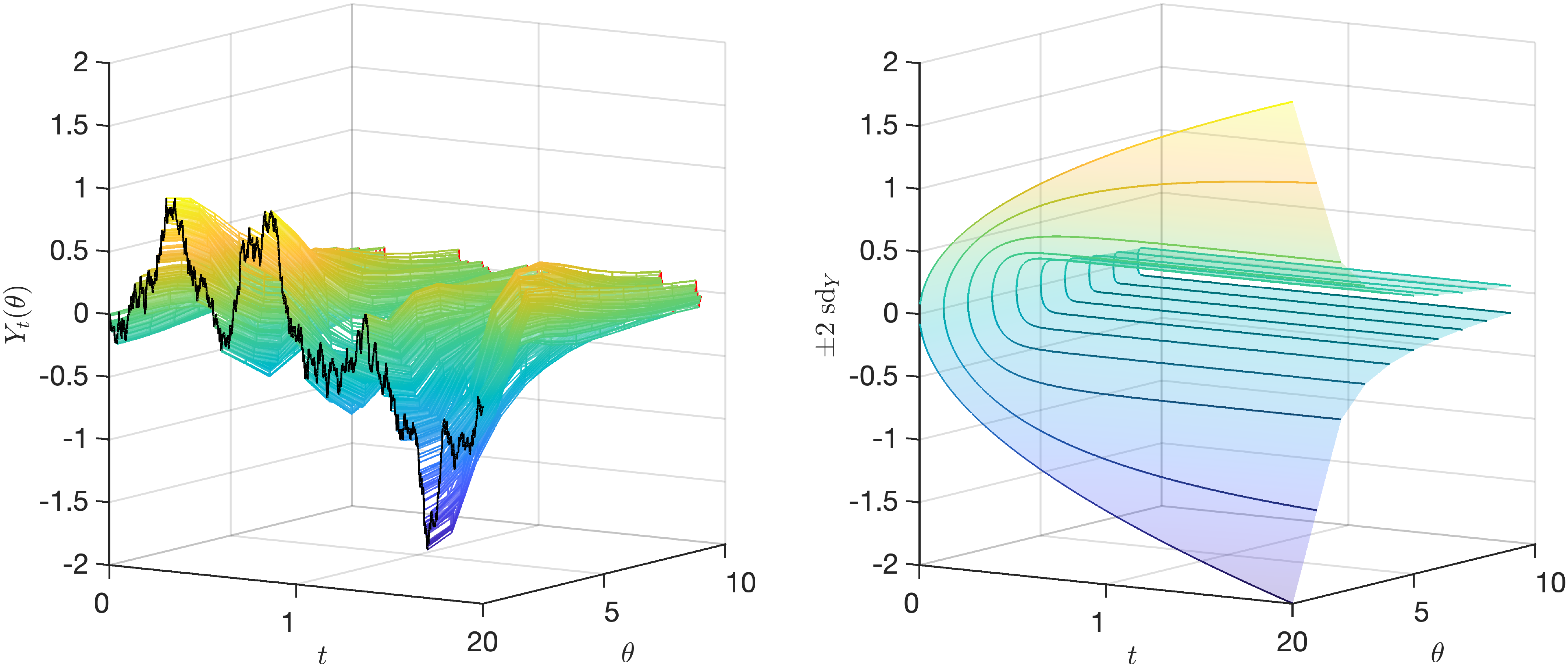}
  \caption{A sample of $Y_t(\theta)$ in \cref{eq:ou-extended-var}, at
    left, with speed of mean reversion $\theta^{2/(1-2H)}$ for
    $H=0.07$ plotted together with an envelope demonstrating
    plus/minus two standard deviations, at right, (cf.~time series
    data of asset prices and option derived price data indicate that
    $H$ often takes values close to $0.1$ or even smaller
    \cite{GatheralEtAl:2018vr}). $Y_t(\theta)$ is a smooth analytic
    function of $\theta$ and discretizing in $\theta$ yields an
    extended variable state space which we utilize in our analysis.}
  \label{fig:Y-t-theta}
\end{figure}

\begin{lemma}[Approximate affine representation]
  \label{lem:approx-affine}
  For $0 < H <1/2$, let
\begin{equation}
  \label{eq:affine-approx}
  \widehat{W}^H_t
  = c_H \sum_{l=1}^{N_L} Y^l_t \Delta \theta_l \eqdef \mathcal{S}(\bY_t) \,,
\end{equation}
depend on $N_L$ degrees of freedom $\bY_t = (Y^1_t, \dots, Y^{N_L}_t)$
where $Y^l_t \defeq Y_t(\theta_l)$ are OU process in
\cref{eq:ou-extended-var} with speed of mean-reversion $\theta^{p}_l$
for $p$ in \cref{eq:rev-pow}. Then $\widehat{W}^H$ converges to $W^H$
as $L,N_L \to \infty$ in $L^2\left(\Omega; C([0,T]) \right)$.
\end{lemma}

\begin{proof}
We obtain an approximate affine representation of the fBm by
discretizing \cref{eq:affine-rep-singularity-removed} in two steps.
First, we divide the integral in
\cref{eq:affine-rep-singularity-removed} into two parts,
\begin{equation*}
  \label{eq:affine-rep-RL}
  W^H_t =   c_H \int_0^L Y_t(\theta) \dd \theta %
  + \underbrace{c_H \int_L^\infty Y_t(\theta) \dd \theta}_{\eqdef R_L(Y_t)}
  \,,
\end{equation*}
where $R_L$ denotes the error in restricting the integral to a fixed
computational domain $L > 1$. Second, we consider a quadrature rule
\begin{equation*}
  W^H_t = c_H \sum_{l=1}^{N_L} Y_t(\theta_l) \Delta \theta_l
  + R_{N_L}(Y_t) + R_L(Y_t)\,,
\end{equation*}
with points 
$0 \leq \theta_1 < \dots < \theta_{N_L} \leq L$ and weights
$\Delta \theta_l = \theta_{l+1} - \theta_l$ where $R_{N_L}$ denotes
the quadrature truncation error.

That $\widehat{W}^H$ converges to $W^H$ in the limit of $L, N_L$
essentially follows from the `strong rates' in \cite{Harms:2019rb}.
The $R_{N_L}$ can be made arbitrarily small as $Y_t(\theta)$ is a
smooth bounded (even analytic) function of $\theta$ (e.g.~see
\cref{fig:Y-t-theta}), i.e.~the regularity in $\theta$ allows one to
approximate efficiently using arbitrarily higher order quadrature
rules if desired (%
\cite{Harms:2019rb}). The $R_L(Y_t)$ is a mean zero Gaussian process,
since for all $\delta \in [0, H)$,
\begin{equation*}
  \sup_{L\in[1,\infty]} L^\delta \Bigl\|\sup_{t\in[0,T]}|R_L(Y_t)|\Bigr\|_{L^p(\Omega)}
  < \infty \,,
\end{equation*}
guarantees integrability (\cite[Lemma 1(b)]{Harms:2019rb}). Then $R_L$
can be made arbitrarily small for sufficiently large $L$ by observing
that the variance,
\begin{equation*}
  \label{eq:estimate-var-RL}
    \var[R_L(Y_t(\theta))] = c_H^2 \int_L^\infty \!\!\!\int_L^\infty
    \cov(Y_t(\theta),
    Y_t(\eta)) \dd \theta \dd \eta
    \leq c_H^2 \frac{2\pi}{4} \int_L^\infty \theta^{1-p} \dd \theta
    = c_H^2 \frac{\pi}{2} \frac{L^{2-p}}{p-2}\,,
\end{equation*}
decays in $L$ since $p>2$.
\end{proof}

We split the weak error \cref{eq:weak-rate-formal} using the
representations in \cref{lem:affine-rep,lem:approx-affine},
\begin{equation}
  \label{eq:weak-error-quad-and-euler}
  \begin{split}
    \E \big[ \varphi(X_T(W^H)) -
    \varphi(X_T(\widehat{W}^H))\big] +  \E \big[\varphi(X_T(\widehat{W}^H)) -
    \varphi(\bar{X}_T^{\Delta 
      t}(\widehat{W}^H))\big]\\
    +  \E \big[
    \varphi(\bar{X}_T^{\Delta t}(\widehat{W}^H))
    - \varphi(\bar{X}_T^{\Delta t}(W^H))\big]
  \end{split}
\end{equation}
where we emphasize the dependence of the underlying on the driving
process. The first and third terms both correspond to the error in
approximating $W^H$ with $\widehat{W}^H$ and therefore vanish by
\cref{lem:approx-affine}. Indeed, we have the following result.
\begin{lemma}
  \label{lem:weak-convergence-terms13}
  Assume that $\varphi$ and $\psi$ are Lipschitz, the latter uniformly in
  time, with Lipschitz constants $K_\varphi$ and $K_\psi$, respectively. Then
  the first and the third term of~\eqref{eq:weak-error-quad-and-euler}
  converge to zero as $N_L, L \to \infty$. Regarding the third term, the
  convergence is uniform with respect to $\Delta t$.
\end{lemma}
\begin{proof}
  We consider the third term first. By basic probabilistic estimates using the
  Lipschitz property of $\varphi$ and $\psi$, we have
  \begin{align*}
    \abs{\E \big[
    \varphi(\bar{X}_T^{\Delta t}(\widehat{W}^H))
    - \varphi(\bar{X}_T^{\Delta t}(W^H))\big]} &\le K_\varphi \E\left[ \left(
    \bar{X}_T^{\Delta t}(\widehat{W}^H) - \bar{X}_T^{\Delta t}(W^H) \right)^2
                                                 \right]^{1/2}\\
    &= K_\varphi \E\left[ \left( \sum_{i=0}^{n-1} \left( \psi(s_i,
      \widehat{W}^H_{s_i}) - \psi(s_i, W^H_{s_i}) \right) W_{s_i,s_{i+1}}
      \right)^2 \right]^{1/2} \\
    &= K_\varphi \left( \sum_{i=0}^{n-1} \E \left[  \left( \psi(s_i,
      \widehat{W}^H_{s_i}) - \psi(s_i, W^H_{s_i}) \right)^2 \right] (s_{i+1} - s_i)
      \right)^{1/2}\\
    &\le K_\varphi K_\psi \left( \sum_{i=0}^{n-1} \E \left[  \left(
      \widehat{W}^H_{s_i} - W^H_{s_i} \right)^2 \right] (s_{i+1} - s_i) 
      \right)^{1/2}\\
    &\le K_\varphi K_\psi \norm{\widehat{W}^H - W^H}_{L^2(\Omega; C([0,T])}
      T^{1/2} \to 0
  \end{align*}
  as $N_L, L \to \infty$ by Lemma~\ref{lem:approx-affine}. The result for the
  first term follows in the same manner.
\end{proof}

\begin{remark}[Convergence rates in $N_L$, $L$]
  \label{rem:waek-con-term13}
  Following \cite[Theorem 1]{Harms:2019rb}, convergence rates in $N_L$ and
  $L$ for the first and third terms of~\eqref{eq:weak-error-quad-and-euler}
  could undoubtedly be established. Keep in mind, however, that we only use
  the scheme $\bar{X}^{\Delta t}_T(\widehat{W}^H)$ as a tool for the analysis of
  the scheme $\bar{X}^{\Delta t}_T(W^H)$, i.e., with exact simulation of
  $W^H$. Consequently, rates of the convergence in $N_L$ and $L$ are not
  required to get rates of convergence of $\bar{X}^{\Delta t}_T(W^H)$ in terms
  of $\Delta t$. Indeed, with respect to the actual scheme $\bar{X}^{\Delta
    t}_T(W^H)$ analyzed in this paper, the error contributions from the first
  and third terms vanish.
\end{remark}

The sole remaining term in~\eqref{eq:weak-error-quad-and-euler},
\begin{equation}
  \label{eq:weak-error-euler}
  \Err(T,\Delta t) \defeq \E[ \varphi(X_T(\widehat{W}^H))] -
  \E [\varphi(\bar{X}^{\Delta t}_T(\widehat{W}^H))]\,,
\end{equation}
that gives the weak error in the Euler scheme, depends on the
approximate affine representation in \cref{lem:approx-affine}. Indeed,
suppose that we are given an error tolerance $\varepsilon$. By
Lemma~\ref{lem:weak-convergence-terms13}, we can find
$L = L(\varepsilon, H, T, \varphi, \psi)$ and
$N_L = N_L(\varepsilon, H, T, \varphi, \psi)$ such that the first and
third terms of~\eqref{eq:weak-error-quad-and-euler} are bounded by
$\varepsilon/3$ each, the third one irrespectively of $\Delta t$. Our
task is now to choose time steps $\Delta t$ such that also the second
term is bounded by $\varepsilon/3$ for the given $L, N_L$. In the next
section, we will obtain an extended variable system for the dynamics
of the underlying that we will use to obtain an asymptotic expansions
for \cref{eq:weak-error-euler}.

\begin{remark}[Quadrature]
  In the interest of keeping our arguments constructive, we first fixed a
  computational domain $L$ and then introduced a quadrature based on $N_L$
  points without specifying the precise rule. One could also choose, e.g., a
  Gauss--Laguerre quadrature suitable for the half-line thereby reducing the
  number of parameters to one (see also \cite{Harms:2019rb}). The splitting
  \cref{eq:weak-error-euler} still holds.
\end{remark}

\subsection{Forward Euler scheme for extended variable system}
\label{sec:euler-scheme}

Substituting \cref{eq:affine-approx} into the underlying
\cref{eq:underlying}, yields
\begin{equation}
  \label{eq:underlying-approx}
  \widehat{X}_t \defeq
   \int_0^t \psi\bigl(s,  \widehat{W}^H_s \bigr) \dd W_s
  = \int_0^t \psi\bigl(s,  \mathcal{S}(\bY_s) \bigr) \dd W_s \,,
\end{equation}
a finite dimensional Markovian approximation
$\widehat{X}_t = X_t(\widehat{W}^H)$ of the underlying $X_t$ that
appears in the weak error \cref{eq:weak-error-euler}. The dynamics of
\cref{eq:underlying-approx} are described by
\begin{equation*}
  \bZ = (\widehat{X}, Y^1, \dots, Y^{N_L})\,,
\end{equation*}
a $d$-dimensional extended variable state space ($d = {N_L} + 1$),
solving the system
\begin{equation}
  \label{eq:extended-system}
  \dd \bZ_t = - b(\bZ_t, t) \dd t + \sigma(\bZ_t, t) \dd W_t\,,
  \qquad \bZ_0 = 0\,,
\end{equation}
with $d$-vectors $b$ and $\sigma$ given by,
\begin{align*}
  b(\bZ_t, t)
  &\defeq  (0, Y^1_t \theta_1^{p}, \dots, Y^{N_L}_t
    \theta_{N_L}^{p})\,,\\%
  \sigma(\bZ_t, t)
  &\defeq  \left(\psi\bigl(t, \mathcal{S}(\bY_t)\bigr), 1, \dots, 1 \right)\,, %
\end{align*}
where there is a single driving Brownian motion, i.e.~
\cref{eq:extended-system} is a degenerate system.

For the interval $[0,T]$, we define the uniform time grid
discretization $t_i \defeq i \Delta t$ for $i = 0, \dots, n-1$, where
$n \defeq T/\Delta t$, and consider the Euler--Maruyama scheme for the
underlying,
\begin{equation}
  \label{eq:euler-scheme}
  \begin{aligned}
    \bar{X}_{t_{i+1}} 
    &=\bar{X}_{t_i} + \psi\bigl(t_i, \mathcal{S}(\bY_{t_i})\bigr)
    \Delta W_{t_i}\,,%
    \quad \bar{X}_{t_0} = 0\,,
  \end{aligned}
\end{equation}
where
\begin{equation*}
  \Delta W_{t_i} \defeq W_{t_i, t_{i+1}} = W_{t_{i+1}}-W_{t_i}
\end{equation*}
are the increments of the driving Wiener process and where at each
time step the vector of extended variables
$\bY_{\cdot} = (Y_{\cdot}^l)_{l=1,\dots, {N_L}}$ is sampled exactly.
That is, for the Euler update at $t_{i+1}$, one can form the
joint distribution
\begin{equation}
  \label{eq:direct-sampling-joint}
  ( \bY_\tau, \Delta W_\tau)_{\tau = t_0, \dots, t_i}\,,
\end{equation}
an $(N_L+1)\times (i+1)$-dimensional Gaussian. The variance-covariance
matrix for \cref{eq:direct-sampling-joint} can be obtained using the
known covariances $\cov(Y^k_{t_i}, Y^l_{t_j})$,
$\cov(Y^k_{t_i}, \Delta W_{t_j})$, and
$\cov(\Delta W_{t_i}, \Delta W_{t_j})$, and then the target variables
$\bY_{t_i}$ required in \cref{eq:euler-scheme} can be sampled exactly
from the joint distribution, e.g., using the Cholesky decomposition.
We extend $\bar{X}$ in \cref{eq:euler-scheme} to all $t \in [0,T]$ by
the interpolation,
\begin{equation}
  \label{eq:euler-interp}
  \begin{aligned}
    \bar{X}(t) =
    \int_{0}^{t}\psi\bigl(\kappa_s, \widehat{W}^H_{\kappa_s} \bigr) \dd{W_s} 
    = \int_0^t \psi\Bigl(\kappa_s, c_H \sum_{l=1}^{N_L}
    \Delta\theta_l
    Y^l_{\kappa_s} \Bigr) \dd W_s \,,
  \end{aligned}
\end{equation}
where $\kappa_s = t_i$ if $s \in [t_i, t_{i+1})$ for each
$i = 0, \dots, n-1$.

Coupling the interpolant for the Euler scheme with the exact dynamics
of the OU variables leads us to define the `discretized' extended
variable system $\bar{\bZ} = (\bar{X}, Y^1, \dots, Y^{N_L})$ embedded
in the SDE
\begin{equation}
  \label{eq:euler-sde}
  \dd \bar{\bZ}_s = - \bar{b}(\bar{\bZ}_s) \dd{s} 
  + \bar{\sigma}(\bar{\bZ}_s)\dd{W_s}\,, \quad s \in [0,T]\,,
  \qquad \bar{\bZ}_0 = 0\,,
\end{equation}
with coefficients
\begin{equation*}
  \label{eq:coeff-euler-sde}
  \begin{split}
    \bar{b}(\bar{\bZ}_s) &= (0, Y^1_{s} \theta_1^{p}, \dots,
    Y^{N_L}_{s}\theta_{N_L}^{p}) = b(\bZ_{s},
    s)\,,\\%
    \bar{\sigma}(\bar{\bZ}_s) &= \left(\psi\bigl(\kappa_s, 
      \mathcal{S}(\bY_{\kappa_s})\bigr), 1, \dots, 1 \right) =
    \sigma(\bZ_{\kappa_s}, \kappa_s)\,.%
  \end{split}
\end{equation*}
For \cref{eq:euler-sde,eq:coeff-euler-sde}, we are able to formulate a
corresponding Kolmogorov backward equation. For a smooth and bounded
payoff $\varphi(\bZ_T) = \varphi(Z^1_T) =
\varphi(\widehat{X}_T)$, %
we consider the value function,
\begin{equation}
  \label{eq:value-function}
  u(\bz,t) \defeq \E[\varphi(\bZ_T) \given \bZ_t = \bz]
  = \E[\varphi(\widehat{X}_T) \given \bZ_t = \bz]\,,
\end{equation}
that is the conditional expected value of the payoff at time $t < T$
given the starting value $\bZ_t = \bz$ for
$\bz = (z_1, \dots, z_d) = (x, y_1, \dots, y_{N_L})$. The associated
Kolmogorov backward equation is given by
\begin{equation}
  \label{eq:Kolmogorov-backward-eqn}
  \begin{split}
  &\partial_t u(\bz,t) - b^j(\bz, t) D_j u(\bz,t)
  + \frac{1}{2} A^{jk}(\bz, t) D_{jk} u(\bz,t) = 0\,, \quad t < T \,, \quad \bz \in \rset^d\,,\\
  &u(\bz,T) = \varphi(z_1) \,,
\end{split}
\end{equation}
where repeated indices indicate summation (over $1, \dots, d$), 
$b^j = b^j(\bz, t)$ is the $j$th component of the $d$-vector,
\begin{equation*}
  \label{eq:extended-var-drift}
  b(\bz, t) \defeq (0, z_{2} \theta_1^{p}, \dots, 
  z_{N_L+1} \theta_{N_L}^{p})\,,
\end{equation*}
and $A^{jk} = A^{jk}(\bz, t)$ are elements of the $d \times d$-matrix
$A = (\sigma \sigma^*)$,
\begin{equation}
  \label{eq:extended-var-diffusion}
  A^{11} = \psi(t, \mathcal{S}(\bz))^2\,,\qquad
  A^{1j}= A^{j1}  = \psi(t, \mathcal{S}(\bz))\,, \; j > 1\,, \qquad
  A^{jk} = 1\,, \; j,k > 1\,,
\end{equation}
i.e.~that contains ones except along the first row
and column, where we slightly abuse notation,
\begin{equation*}
  \label{eq:sum-extended-vars}
  \mathcal{S}(\bz)  \defeq 0\cdot z_1 + c_H \sum_{j=2}^{d} z_{j}
  \Delta \theta_{j+1} = c_H \sum_{l=1}^{N_L} y_{l}
  \Delta \theta_l\,.
\end{equation*}

\begin{remark}[Kolmogorov backward equation]
  In the presentation of the Kolmogorov backward equation, we assume
  necessary regularity conditions on $\varphi$, i.e.~smoothness and
  boundedness, as a matter of convenience. From the context of the
  problem this is not a strong assumption, see
  \cref{rem:assumptions-main-theorem}.
\end{remark}

\subsection{Local weak error representation}
\label{sec:local-weak-error}

Throughout the remainder of this work we consider the case when
$\psi(s, W^H_s) = W^H_s$ in \cref{eq:underlying}.
Returning to \cref{eq:weak-error-euler}, we obtain a representation
for the weak error in terms of local errors. First, we write the
discretization error as a telescoping sequence in the value function
\cref{eq:value-function},
\begin{align}
  \Err(T, \Delta t)
  &= \E \bigl[\varphi(\widehat{X}_T) - \varphi(\bar{X}_{t_n})\bigr]
    \notag \\
  &= - \left( \E u(\bar{\bZ}_{t_n}, T) - \E u(\bar{\bZ}_{t_0}, 0)\right)
    \notag\\
  &= - \sum_{i=0}^{n-1} \E \left[ u(\bar{\bZ}_{t_{i+1}}, t_{i+1})
    - u(\bar{\bZ}_{t_i}, t_i)\right]\,,
    \label{eq:telescoping-sum-value-fcn}
\end{align}
using that
\begin{equation*}
  \E\varphi(\bar{X}_{t_n}) = \E\big[\E[\varphi(Z^1_{T})
  \given \bZ_T = \bar{\bZ}_{t_n}]\big]
  = \E u(\bar{\bZ}_{t_n}, T) 
\end{equation*}
and
\begin{equation*}
  \E\varphi(\widehat{X}_T)
  = \E\big[\E[\varphi(Z^1_T)
  \given \bZ_{0} = \bar{\bZ}_{t_0}]\big]
  = \E u(\bar{\bZ}_{t_0}, 0) \,.
\end{equation*}
We then represent each difference appearing in
\cref{eq:telescoping-sum-value-fcn} as a stochastic differential over
a small time increment. %
Over the interval $[t_i, t_{i+1})$, we have that
\begin{equation*}
  \begin{aligned}
    \E[u(\bar{\bZ}_{t_{i+1}}, t_{i+1}) - &u(\bar{\bZ}_{t_i}, t_i)]
    = \E \int_{t_i}^{t_{i+1}} \dd u(\bar{\bZ}_s, s)\\
    &= \E\int_{t_i}^{t_{i+1}} \!\!\Big( \partial_t u (\bar{\bZ}_s, s)
    - b^j(\bar{\bZ}_s, s) D_j u(\bar{\bZ}_s, s)+ \frac{1}{2}
    A^{jk}(\bar{\bZ}_{t_i}, t_i) D_{jk} u(\bar{\bZ}_s, s) \Big)\dd s \,,
  \end{aligned}
\end{equation*}
using It\={o}'s formula applied to \cref{eq:euler-sde} %
where repeated indices indicate summation over $1, \dots, d$.
Subtracting off the Kolmogorov backward equation
\cref{eq:Kolmogorov-backward-eqn} evaluated at $(\bar{\bZ}_s, s)$ then
yields
\begin{equation}
  \label{eq:increment-value-func-A}
  \E[u(\bar{\bZ}_{t_{i+1}}, t_{i+1}) - u(\bar{\bZ}_{t_i}, t_i)] = %
  \frac{1}{2} \E \int_{t_i}^{t_{i+1}} \left(A^{jk}(\bar{\bZ}_{t_i}, t_i) -
    A^{jk}(\bar{\bZ}_s, s)\right) D_{jk}u(\bar{\bZ}_s, s) \dd s\,.
\end{equation}
We note that the non-zero terms correspond to differences along the
first row and column of $A$ in \cref{eq:extended-var-diffusion}, and
thus \cref{eq:increment-value-func-A} simplifies to the following
expression for the local weak error in the value function, 
\begin{equation}
  \label{eq:increment-value-func-A-simp}
  \begin{split}
    \E[u(\bar{\bZ}_{t_{i+1}}, t_{i+1}) - u(\bar{\bZ}_{t_{i}}, t_i)] 
    &= \frac{1}{2} \E \int_{t_i}^{t_{i+1}} \left(
      A^{11}(\bar{\bZ}_{t_i}, t_i)
      - A^{11}(\bar{\bZ}_s, s)\right)D_{11}u(\bar{\bZ}_s, s) \dd s\\
    &\qquad + \E \int_{t_i}^{t_{i+1}} \sum_{j=2}^d
    \left(A^{j1}(\bar{\bZ}_{t_{i}}, t_i) -
      A^{j1}(\bar{\bZ}_s, s)\right)D_{j1}u(\bar{\bZ}_s, s) \dd s\\
    &= - \frac{1}{2} \E \int_{t_i}^{t_{i+1}} \DeltaSYSY D_{11}u(\bar{\bZ}_s, s) \dd s\\
    &\qquad - \E \int_{t_i}^{t_{i+1}} \DeltaSY \sum_{j=2}^d
     D_{j1}u(\bar{\bZ}_s, s) \dd s\,,
  \end{split}
\end{equation}
where we express the differences in components of
$A$
in terms of the increments of the approximate fBm
\begin{equation*}
  \label{eq:increment-SY}
  \DeltaSYk \defeq  \mathcal{S}(\bY_s)^k - \mathcal{S}(\bY_{t_i})^k\,,
  \qquad k=1,2\,.
\end{equation*}
Observe that the derivatives of the value function appearing in
\cref{eq:increment-value-func-A-simp} can be further resolved by
directly computing fluxes. Here we consider the simplification
$\psi(s, \widehat{W}^H_s) = \widehat{W}^H_s$ in \cref{eq:underlying}
(i.e.~`linear $\psi$').

\begin{lemma}[Fluxes]
  \label{lem:fluxes}
  Let $\psi(s, W^H_s) = W^H_s$. For the value function $u(\bz, t)$
  defined in \cref{eq:value-function},
  \begin{equation}
    \label{eq:diff-u}
    D^\beta u (\bz, t)
    = c_H^{|\beta| - \beta_1}
    \E \Big[\varphi^{(|\beta|)}(\widehat{X}_T)
    \prod_{j=1}^{N_L} (\Delta\theta_{j} M^j_{t,T})^{\beta_{j+1}}
    \given \bZ_t = \bz \Big]\,,
  \end{equation}
  for a multi-index $\beta = (\beta_1, \dots, \beta_d)$ where
  \begin{equation*}
    M^j_{t,T} \defeq \int_{t}^{T}e^{-(r-t)\theta^p_j}\dd{W_r}\,,
    \quad j=1,\dots, N_L\,. 
  \end{equation*}
\end{lemma}

\begin{proof}
  Let $\bZ^{t, \bz}_s$ be the Markov process started at $\bZ_t = \bz$
  with components
  \begin{equation*}
    \bZ^{t, \bz}_s = (\widehat{X}^{t,x}_s, \bY^{\,t,\boldsymbol{y}}_s) 
    = (\widehat{X}^{t,x}_s, Y^{t, y_1}_s, \dots, Y^{t, y_{N_L}}_s)\,;
  \end{equation*}
  here we drop the index $Y^{t,y_l} = Y^{l; t, y_l}$ when the index is
  clear from the initial condition. We recall that $Y^{t,y_l}_s$
  started at the value $y_l$ at time $t$ is given by,
  \begin{equation}
    \begin{split}
      \label{eq:ou-started-at-y}
      Y^{t, y_l}_s
      = e^{-(s-t)\theta^p_l}y_l + \int_t^s e^{-(s-r)\theta^p_l} \dd
      W_r\,, \quad t < s\,,
    \end{split}
  \end{equation}
  and, similarly, that $\widehat{X}^{t,x}_s$ is given by,
  \begin{equation*}
    \label{eq:underlying-approx-started-at-x}
    \begin{split}
      \widehat{X}^{t,x}_s %
      = x + \int_{t}^{s} \mathcal{S}(\bY^{\,t,\boldsymbol{y}}_r)
      \dd{W_r} = x + \int_{t}^{s} c_H \sum_{l=1}^{N_L}Y^{t,y_l}_r
      \Delta \theta_l \dd W_r\,, \quad t<s\,. 
    \end{split}
  \end{equation*}
  Working directly with
  \cref{eq:ou-started-at-y,eq:underlying-approx-started-at-x},
  derivatives of the underlying with respect to the initial conditions
  are given by
  \begin{equation*}
    \label{eq:flux-X-dx}
    \frac{\partial \widehat{X}^{t,x}_s}{\partial x} = 1\,,
  \end{equation*}
  and, for $l=1,\dots,N_L$,
  \begin{equation*}
    \label{eq:flux-X-dy}
    \frac{\partial \widehat{X}^{t,x}_s}{\partial y_l} =
    c_H \Delta \theta_l \int_t^s e^{-(r-t)\theta^p_l} \dd W_r
    \eqdef c_H \Delta\theta_l M^l_{t,s}\,.
  \end{equation*}
  The formula \cref{eq:diff-u} follows as all higher derivatives of
  $\widehat{X}^{t,x}_s$ vanish.
\end{proof}

Returning to our expression for the local weak error in the value
function \cref{eq:increment-value-func-A-simp} we apply
\cref{eq:diff-u} thereby obtaining,
\begin{equation}
  \label{eq:increment-value-func-condexp}
  \begin{split}
    \E[u(\bar{\bZ}_{t_{i+1}}, t_{i+1}) - u(\bar{\bZ}_{t_{i}}, t_i)] =
    - \frac{1}{2} \E \int_{t_i}^{t_{i+1}} \DeltaSYSY
    \E[\varphi''(\widehat{X}_T) \given \bZ_s = \bar{\bZ}_{s}]\dd s \\
     - \sum_{l=1}^{N_L} \E \int_{t_i}^{t_{i+1}} \DeltaSY \E[
    \varphi''(\widehat{X}_T) c_H \Delta\theta_l M^l_{s,T}
    \given \bZ_s = \bar{\bZ}_{s}] \dd s\,.
  \end{split}
\end{equation}
We introduce deterministic functions of $\bz$,
\begin{equation*}
  \label{eq:def-nu}
  \nu(\bz, s) \defeq
  \E [\varphi''(\widehat{X}_T) \given \bZ_s = \bz]\,,
\end{equation*}
and
\begin{equation*}
  \label{eq:def-nu-tilde}
  \widetilde{\nu} (\bz, s) \defeq
  \E \Bigl[\varphi''(\widehat{X}_T)
  \bigl( c_H \textstyle \sum_{l}  M^l_{s,T} \Delta \theta_l \bigr)
  \given \bZ_s = \bz\Bigr]\,.
\end{equation*}
Rewriting \cref{eq:increment-value-func-condexp}
with this new notation leads to an expression for the local weak error in the value function,
\begin{equation}
  \label{eq:increment-value-func}
  \begin{split}
    \E[u(\bar{\bZ}_{t_{i+1}}, t_{i+1}) - u(\bar{\bZ}_{t_{i}}, t_i)] &=
    - \frac{1}{2} \E \int_{t_i}^{t_{i+1}} \!\!\!\DeltaSYSY \nu(\bar{\bZ}_s, s)
    \dd{s} - \E \int_{t_i}^{t_{i+1}}\!\!\!\DeltaSY \widetilde{\nu}(\bar{\bZ}_s,
    s)
    \dd{s} \\
    &\eqdef J + \widetilde{J} \,,
\end{split}
\end{equation}
that will serve as our starting point for the convergence
rates. 
Next we derive an expansion for \cref{eq:increment-value-func} in
powers of $\Delta t$ from which we obtain convergence rates.

\subsection{Taylor expansions and conditional independence}
\label{sec:taylor-expansion}

Starting with the local weak error \cref{eq:increment-value-func}, we
derive asymptotic expansions for $J$ (and $\widetilde{J}$) in in powers of
$\Delta t$ by Taylor expanding $\nu$ (and $\widetilde{\nu}$) at
$\bar{\bZ}_{t_i}$ and applying a conditioning argument.

We observe that $\bar{\bZ} = (\bar{X}, \bY\,)$ in \cref{eq:euler-sde},
i.e.~ the interpolation \cref{eq:euler-interp} together with the exact
dynamics of the OU extended variables, is linear with respect to the
increment over $[t_i, s]$,
\begin{equation}
  \label{eq:linearization-Z}
  \bar{\bZ}_s - \bar{\bZ}_{t_i}
  = \bigl(\mathcal{S}(\bY_{t_i}) W_{t_i, s}\,, Y^1_{t_i, s}\,,
  \dots\,, Y^{N_L}_{t_i, s}  \bigr) 
\end{equation}
where
\begin{equation*}
  W_{t_i, s} \defeq W_s - W_{t_i}  \quad \text{and}
  \quad Y^l_{t_i, s} \defeq Y^l_s - Y^l_{t_i} \,,
  \qquad \text{for } s \in [t_i, t_{i+1})\,,
\end{equation*}
are increments of the driving Brownian motion and extended variable OU
processes, respectively. Using the linearization
\cref{eq:linearization-Z}, the Taylor expansion of $\nu$ at
$\bar{\bZ}_{t_i}$ is,
\begin{equation}
  \label{eq:taylor-nu}
  \nu(\bar{\bZ}_s,s)  = \nu(\bar{\bZ}_{t_i}, s)
  + \sum_{\beta \in \mathcal{I}_{\kappa}} D^\beta \nu(\bar{\bZ}_{t_i}, s)
  \cdot (\mathcal{S}(\bY_{t_i}) W_{t_i,s} )^{\beta_1} (\bY_{t_i,s})^{\hat{\beta}} + \mathcal{R}_\kappa (\nu) \,,
\end{equation}
for sums over multiindices in the set
\begin{equation*}
  \mathcal{I}_{\kappa} = \bigl\{\beta = (\beta_1,\hat{\beta})
  = (\beta_1, \beta_2, \dots, \beta_d) \;:\;
  1 \leq |\beta| \leq \kappa-1 \bigr\}
  \label{eq:multiindex-set}
\end{equation*}
where
\begin{equation*}
  (\bY_{t_i,s})^{\hat{\beta}}
  = \prod_{l=1}^{N_L} (Y^{l}_{t_i,s} )^{\beta_{l+1}}
  = (Y^{1}_{t_i,s} )^{\beta_{2}} \cdots (Y^{N_L}_{t_i,s} )^{\beta_{d}}
\end{equation*}
and the remainder is given in integral form by,
\begin{equation}
  \label{eq:taylor-nu-remainder}
  \begin{split}
    \mathcal{R}_\kappa(\nu) = \frac{1}{\kappa !} \sum_{|\beta| = \kappa}
    (\mathcal{S}(\bY_{t_i}) W_{t_i,s})^{\beta_1}
    (\bY_{t_i,s})^{\hat{\beta}} \int_0^1 D^\beta \nu(\xi_\tau, s)
    \dd \tau \,,\\
    \xi_{\tau} \defeq \bar{\bZ}_{t_i} + \tau
    \bigl(\mathcal{S}(\bY_{t_i}) W_{t_i, s}\,, Y^1_{t_i,
      s}\,, \dots\,, Y^{N_L}_{t_i, s} \bigr)\,.
  \end{split}
\end{equation}
In \cref{eq:taylor-nu}, the terms $\nu$ and $D^\beta \nu$
are deterministic functions of the $\mathcal{F}_{t_i}$-measurable
random variable $\bZ_{t_i}$ and are therefore
$\mathcal{F}_{t_i}$-measurable. Analogous expressions hold for
\cref{eq:taylor-nu,eq:taylor-nu-remainder} with $\widetilde{\nu}$ in place
of $\nu$.

Plugging the $\nu$-expansion \cref{eq:taylor-nu} into
\cref{eq:increment-value-func}, we obtain an asymptotic expansion for
$J$,
\begin{equation}
  \label{eq:expansion-J}
  \begin{split}
    2J &= - \E \int_{t_i}^{t_{i+1}} \nu(\bar{\bZ}_{t_i}, s) \E\bigl[ \DeltaSYSY \given \mathcal{F}_{t_i} \bigr] \dd{s} \\
    &\quad - \sum_{\beta \in \mathcal{I}_\kappa} \E \int_{t_i}^{t_{i+1}}  D^\beta \nu(\bar{\bZ}_{t_i}, s) \mathcal{S}(\bY_{t_i})^{\beta_1} \E\bigl[ \DeltaSYSY (W_{t_i,s})^{\beta_1} (\bY_{t_i,s})^{\hat{\beta}} \given \mathcal{F}_{t_i}  \bigr] \dd{s} \\
    &\quad - \E \int_{t_i}^{t_{i+1}} \E\bigl[ \DeltaSYSY
    \mathcal{R}_\kappa(\nu) \given \mathcal{F}_{t_i} \bigr] \dd{s}\,,
  \end{split}
\end{equation}
by conditional independence. An expansion analogous to
\cref{eq:expansion-J} holds for $\widetilde{J}$ with $\widetilde{\nu}$ in
place of $\nu$. The only terms that depend on $\Delta t$ in
\cref{eq:expansion-J} (and in the analogously expansion for
$\widetilde{J}$) are the conditional expectations involving products of
the increments $\DeltaSYSY$ or $\DeltaSY$ together with powers of
$W_{t_i,s}$ and $\bY_{t_i,s}$. The key obtain weak error rates will be
to show after isolating the order in $\Delta t$ using these expansions
that the expansion coefficients, which depend on the extended state
space variables $\bY$, are controlled with respect to summation in the
parameter(s) $\theta$.

In the next section, we observe that for quadratic payoffs, the
expansions in $\nu$ and $\widetilde{\nu}$ truncate after the first term
since $\nu$ and $\widetilde{\nu}$ already depend on two derivatives of
$\varphi$. In this special case, we derive weak rate one in
\cref{thm:weak-rate-quad}. In \cref{sec:weak-rate-gen}, we prove that
in general the weak rate is $H+1/2$, as reported in
\cref{thm:weak-rate-gen}, also using the asymptotic expansions
approach.

\section{Weak rate one for quadratic payoffs}%
\label{sec:weak-rate-quad}

Using the preceding machinery, we will now derive rates of convergence
for the weak error via Taylor expansions in powers of $\Delta t$ such
that all terms stay integrable in $\theta$. For quadratic payoff
functions, we obtain that the weak error is $O(\Delta t)$, i.e.~rate
one in \cref{thm:weak-rate-quad} below, which is supported by
numerical evidence, recall \cref{fig:conv-quadratic-payoff}. The
mechanism by which rate one is achieved can be observed in the
expansions; the expansion coefficients depend on derivatives of the
payoff function and higher-order terms that reduce the rate vanish
when $\varphi$ is quadratic.

\subsection{Asymptotic expansion approach to weak rate one}
\label{sec:weak-rate-quad:one}

Returning to the increment of the value functional
\cref{eq:increment-value-func}, if $\varphi \in \mathcal{P}^2$, that
is, is a quadratic polynomial, then the derivatives of $\nu$ and
$\widetilde{\nu}$ (as defined in \cref{eq:def-nu,eq:def-nu-tilde}) vanish
and only the first terms in the expansion \cref{eq:taylor-nu} remain.
Then,
\begin{equation}
  \label{eq:splitting-quad}
  \begin{split}
    J + \widetilde{J} &= - \frac{1}{2} \E \int_{t_i}^{t_{i+1}}
    \nu(\bar{\bZ}_{t_i}, s) \E \bigl[ \DeltaSYSY \given
    \mathcal{F}_{t_i} \bigr] \dd{s} - \E
    \int_{t_i}^{t_{i+1}}\widetilde{\nu}(\bZ_{t_i}, s)
    \E\bigl[ \DeltaSY \given \mathcal{F}_{t_i} \bigr] \dd{s} \\
    &\eqdef \tfrac{1}{2} J_{0} + \widetilde{J}_{0}\,,
  \end{split}
\end{equation}
and estimating $J_{0}$ and $\widetilde{J}_{0}$ yields the weak rate
corresponding to quadratic payoff functions.

\begin{theorem}[Weak rate quadratic payoff]
  \label{thm:weak-rate-quad}
  Let $\varphi \in \mathcal{P}_2$ and let
  $\psi(s, \widehat{W}^H_s) = \widehat{W}^H_s$, then
  \begin{equation*}
    \label{eq:weak-rate-one-quad}
    \Err(T, \Delta t, \varphi)
    \defeq \E \bigl[\varphi(\widehat{X}_T) -  \varphi(\bar{X}_{t_n})\bigr]
    \lesssim O(\Delta t)\,,
  \end{equation*}
  i.e.~the Euler method is weak rate one.
\end{theorem}

\begin{proof}
  We estimate the terms $J_{0}$ and $\widetilde{J}_{0}$ in
  \cref{eq:splitting-quad} beginning with $\widetilde{J}_{0}$. Working
  directly with the increments of the extended variables,
  \begin{equation*}
    \E\bigl[ \DeltaSY \given \mathcal{F}_{t_i} \bigr]
    = - c_H \sum_{l=1}^{N_L} \E\bigl[ Y^l_{t_i,s} \given \mathcal{F}_{t_i} \bigr] \Delta \theta_l
    = 0\,,
  \end{equation*}
  and thus we conclude
  \begin{equation*}
    \widetilde{J}_{0} = - \E \int_{t_i}^{t_{i+1}} \widetilde{\nu}(\bar{\bZ}_{t_i}, s)
    \E\bigl[ \DeltaSY \given \mathcal{F}_{t_i} \bigr] = 0\,.
  \end{equation*}

  In the case of quadratic $\varphi$, we observe that
  $\|\nu\|_{\infty} = O(1)$ and deterministic. From the definition of
  $Y_s$ (e.g.~see \cref{eq:ou-started-at-y}), working again directly
  with the increments of the extended variables we have that
  \begin{equation*}
    \label{eq:DeltaSYSY-expand}
    \begin{split}
      \DeltaSYSY &= - c_H^2 \sum_{k,l=1}^{N_L}(Y^k_{t_i}Y^l_{t_i} -
      Y^k_{s}Y^l_{s})
      \Delta \theta_k \Delta \theta_l\\
      &= - c_H^2 \sum_{k,l=1}^{N_L} \Bigl\{
      (1-e^{-(s-t_i)(\theta^p_k+\theta^p_l)})Y^k_{t_i}Y^l_{t_i} -
      e^{-(s-t_{i})\theta^p_k}Y^k_{t_i}
      \int_{t_i}^{s}\!\!e^{-(s-r)\theta^p_l}\dd{W_r} \\
      &\qquad - e^{-(s-t_{i})\theta^p_l}Y^l_{t_i}
      \int_{t_i}^{s}\!\!e^{-(s-r)\theta^p_k}\dd{W_r} -
      \int_{t_i}^{s}\!\!e^{-(s-r)\theta^p_k} \dd{W_r}
      \int_{t_i}^{s}\!\!e^{-(s-r)\theta^p_l} \dd{W_r} \Bigr\} \Delta
      \theta_k \Delta \theta_l \,.
    \end{split}
  \end{equation*}
  From this the key conditional expectation term reduces to
  \begin{align}
    \E\bigl[ \DeltaSYSY \given \mathcal{F}_{t_i} \bigr]
      &= - c_H^2 \sum_{l,k=1}^{N_L}\Delta\theta_l\Delta\theta_k
        (1-e^{-(s-t_i)(\theta^p_l+\theta^p_k)})
        \bigl(Y^{l}_{t_i}Y^{k}_{t_i} - \tfrac{1}{\theta^p_l + \theta^p_k}\bigr)
        \label{eq:J0-cond-expec}\,.
  \end{align}
  Only the component $1 - e^{- \Delta t(\theta^p_l+\theta^p_k)}$ will
  contribute to the estimate for the weak rate, provided that the sums
  over $\theta$ in \cref{eq:J0-cond-expec} converge.

  Using \cref{eq:J0-cond-expec} we find
  \begin{equation}
    \label{eq:J0-first-term-naive}
    \begin{split} %
      J_{0} &= -2 c_H^2 \int_{t_i}^{t_{i+1}}  \sum_{l,k=1}^{N_L} \Delta
      \theta_l \Delta \theta_k \E [Y^l_{t_i}Y^k_{t_i} -
      \tfrac{1}{\theta^p_l + \theta^p_k}]
      (1-e^{-(s-t_i)(\theta^p_l+\theta^p_k)})  \dd s\\
      &= -2 c_H^2 \sum_{l,k=1}^{N_L} \Delta\theta_l \Delta\theta_k \E
      \bigl[ Y^l_{t_i}Y^k_{t_i} - \tfrac{1}{\theta^p_l+ \theta^p_k}
      \bigr] \int_{0}^{\Delta t}   g(s) \dd{s}\,.
    \end{split}
  \end{equation}
  The integrand,
  \begin{equation}
    \label{eq:lip-func-naive}
    g(s)
    \defeq 1-e^{-s(\theta^p_l+\theta^p_k)}\,,
  \end{equation}
  is a function of $s \in [0, \Delta t)$ such that $g(0) = 0$ and the
  associated Lipschitz constant $K$ is given by,
  \begin{equation*}
    \label{eq:lip-const-naive}
    K = \frac{\partial}{\partial s}g(s) \Big|_{s = 0}
    = \theta^p_l+\theta^p_k \,.
  \end{equation*}
  Then, since $|g(s) - g(0)| \leq K s$, we have
  \begin{equation*}
    |J_{0}|
    \leq 2 c_H^2 \sum_{l,k=1}^{N_L} \Delta\theta_l \Delta\theta_k
    \left| \E \bigl[ Y^l_{t_i}Y^k_{t_i}(\theta^p_l+\theta^p_k)
      - 1 \bigr] \right|  
    \int_0^{\Delta t} s \dd{s}\,.
  \end{equation*}
  Computing the covariance appearing above,
  \begin{equation*}
    \label{eq:covar-Yl-Yk}
    \E [Y^l_{t_i}Y^k_{t_i}]
    =  \int_{0}^{t_i} e^{-(t_i-r)(\theta^p_l+\theta^p_k)}\dd r
    = \frac{1-e^{-t_i(\theta^p_l+\theta^p_k)}}{\theta^p_l+\theta^p_k}\,,
  \end{equation*}
  we see that
  \begin{equation}
    \label{eq:est-J0-naive}
    |J_{0}| \leq 2 c_H^2 \Bigl( \sum_{l=1}^{N_L}\Delta\theta_l
    e^{-t_i \theta^p_l} \Bigr)^2  \Delta t^2 \leq  \frac{c_H^2}{p^2}
    \Gamma\bigl(\tfrac{1}{p}\bigr)^2  t_i^{-2/p} \Delta t^2 \lesssim O(\Delta t^2)\,,
  \end{equation}
  since
  \begin{equation*}
    \sum_{l=1}^{N_L} \Delta \theta_l e^{-t_i
      \theta^p_l}\leq \int_0^L e^{-t_i \theta^p_l} \dd \theta_l \leq
    \int_0^\infty e^{-t_i \theta^p_l} \dd \theta_l =
    \frac{1}{p}\Gamma\big(\tfrac{1}{p}\big) t_i^{-1/p}\,.
  \end{equation*}
  Importantly, $t^{-2/p}$ is integrable on $[0,T]$ when $p>2$ and
  therefore the $t_i^{-2/p}$ appearing in \cref{eq:est-J0-naive} will
  remain uniformly bounded when summing over $t_i$ in
  \cref{eq:telescoping-sum-value-fcn}.

  Turning now to the telescoping representation of the weak error
  \cref{eq:telescoping-sum-value-fcn} and using
  \cref{eq:est-J0-naive}, we obtain the desired rate
  \begin{align*} %
    \left|\Err(T,\Delta t,\varphi)\right|
    &= \Bigl|\sum_{i=0}^{n-1}\E\bigl[ u(\bar{\bZ}_{t_{i+1}}, t_{i+1}) -
      u(\bar{\bZ}_{t_i}, t_i)\bigr] \Bigr| \\
    &= \Bigl| \sum_{i=0}^{n-1} \E \int_{t_i}^{t_{i+1}}\nu(\bar{\bZ}_{t_i}, s)
      \E\bigl[ \DeltaSYSY \given \mathcal{F}_{t_i} \bigr] \dd{s} \Bigr|\\
    &\leq \frac{c_H^2}{p^2} \Gamma\bigl( \tfrac{1}{p} \bigr)^2
      \sum_{i=0}^{n-1} \Delta t^2 t_i^{-2/p} \\
    &\leq \frac{c_H^2}{p(p-2)} \Gamma\bigl( \tfrac{1}{p} \bigr)^2
       T^{(p-2)/p} \Delta t\\
    &\leq T^{2H} \Delta t \,. \qedhere
  \end{align*}
\end{proof}

Although initially surprising that the rate depends on the payoff, the
Taylor expansion for the weak error provides insight into this
behavior. The expansion coefficients in \cref{eq:taylor-nu} depend on
increasingly higher order derivatives of the payoff function $\varphi$
through derivatives of $\nu$ and $\widetilde{\nu}$ (in
\cref{eq:def-nu,eq:def-nu-tilde}). Unlike for quadratic $\varphi$
where terms in \cref{eq:taylor-nu} vanish, the higher order derivative
terms persist for general payoff functions. Oddly, it is only the next
higher term (compared to those involved in the expansion for quadratic
$\varphi$) that reduces the overall rate for general $\varphi$. In
this context, one might hope to obtain an \emph{effective rate} that
is independent of $H$ for payoff functions well approximated by
quadratic polynomials.
               
\subsection{A simpler proof for weak rate one}%
\label{sec:simple-proof-rate-one}%

Before moving on to the proof of the main result in
\cref{sec:weak-rate-gen}, we first present a simpler proof of weak
rate one for quadratic payoff functions that is also applicable to
nonlinear $\psi$. This proof as well as the weak rate itself was communicated to us by A.~Neuenkirch \cite{neuenkirch-comm}.

\begin{lemma}
  \label{lem:weak-rate-quad-simple}
  Suppose that $\psi \in C^1_{\mathrm{pol}}$, i.e.~$\psi \in C^1$ and
  $\psi, \partial_t \psi, \partial_x \psi$ have polynomial growth, and $\varphi(x) = x^2$. Then
  \begin{equation*}
    \left| \E\bigl[ \varphi( X_T ) -  \varphi(
      \widetilde{X}_T ) \bigr] \right| = O(\Delta t)\,,
  \end{equation*}
  where 
  \begin{equation*}
    \widetilde{X}_T \defeq \int_0^T \psi(s, W^H_{\kappa_s}) \dd W_s\,,
  \end{equation*}
  for $\kappa_s = t_i$ if $s \in [t_i, t_{i+1})$ for each
  $i = 0, \dots, n-1$.
\end{lemma}
\begin{proof}
  By the It\={o} isometry, we have
  \begin{gather*}%
    \E\left[ \varphi\left( X_T \right) \right] = \int_0^T \E\left[
      \psi^2(s, s^H W^H_1) \right] \dd s = \int_0^T g(s) \dd s\,,\\
    \E \bigl[ \varphi( \widetilde{X}_T ) \bigr] = \int_0^T \E\left[
      \psi^2(\kappa_s, \kappa_s^H W_1^H) \right] \dd s =
    \int_0^T g(\kappa_s) \dd s\,,
  \end{gather*}%
  where
  \begin{equation*}
    g(s) \defeq \E\left[ \psi^2(s,  s^H V ) \right], \quad V \sim
    \mathsf{N}(0,1)\,. 
  \end{equation*}
  Note that $g$ is differentiable with integrable derivative and we assume the time derivative of $\psi$ is bounded. Indeed,
  \begin{equation*}
    g^\prime(t) = 2\E\left[ \partial_t \psi(t, t^H V ) \right]
    + 2\E \left[
      \partial_x \psi(t, t^H V) V \right] t^{H-1}\,,
  \end{equation*}
  which is of order $t^{H-1}$ and, hence, integrable.

  Setting
  \begin{equation*}
    \zeta_t \defeq \min \{t_i \given t_i \ge t\}\,,
  \end{equation*}
  we conclude with
  \begin{align*}
    \abs{\int_0^T g(s) \dd s - \int_0^T g(\kappa_s) \dd s}
    &\le \int_0^T \abs{g(\kappa_s) + \int_{\kappa_s}^s g^\prime(t) \dd t -
      g(\kappa_s)} \dd s\\
    &\le \int_0^T \int_{\kappa_s}^s \abs{g^\prime(t)} \dd t \dd s\\
    &= \int_0^T \int_t^{\zeta_t} \dd s \abs{g^\prime(t)} \dd t\\
    &\le \max_{i=0, \ldots, n-1} \abs{t_{i+1} - t_i}
      \left\lVert g^\prime \right\rVert_{L^1([0,T])}\,. \qedhere
  \end{align*}
\end{proof}

For simplicity, we assumed that
$\psi(s, \widehat{W}_s^H)=\widehat{W}_s^H$ in
\cref{thm:weak-rate-quad}. In contrast
\cref{lem:weak-rate-quad-simple} is applicable to any $\psi$ including
nonlinear functions. However, it is not clear how to extend the
approach of the simple proof for \cref{lem:weak-rate-quad-simple} to
more general payoff functions $\varphi$. In the next section we use
the asymptotic expansions to prove the main result,
\cref{thm:weak-rate-gen}, obtaining the weak rate $H+1/2$ for general
payoff functions for $\psi(s, \widehat{W}_s^H)=\widehat{W}_s^H$.

\section{Proof of \cref{thm:weak-rate-gen}}%
\label{sec:weak-rate-gen}%

Our proof of \cref{thm:weak-rate-gen} follows the expansion approach
used to obtain rate one for quadratic payoffs in
\cref{sec:weak-rate-quad:one}. We derive asymptotic expansions in
powers of $\Delta t$ for increments of the value function in
\cref{eq:increment-value-func}, i.e., for $J$ and $\widetilde{J}$. For the
case of general payoff functions, this requires two rounds of Taylor
expansions. The first round expands $\nu$ and $\widetilde{\nu}$ at
$\bar{\bZ}_{t_i}$ using \cref{eq:taylor-nu}, as was done in
\cref{sec:weak-rate-quad:one}. Then after applying a conditioning
argument, we explicitly deal with correlations by expressing our
expansions in terms of the extended variables and making a second
round of Taylor expansions with respect to select components of
$\bY_{t_i}$. Again, a key point in the proof is that all the terms in
the expansions are controlled with respect to $\theta$.

Returning to the local weak error
\cref{eq:increment-value-func-A-simp}, we use the $\nu$-expansion
\cref{eq:taylor-nu} to find that $J$, the term corresponding to the
increment $\DeltaSYSY$, up to order $\kappa = 3$ is given by,
\begin{align}
  2 J   &= - \E \int_{t_i}^{t_{i+1}} \nu(\bar{\bZ}_{t_i}, s)
            \E\bigl[ \DeltaSYSY \given \mathcal{F}_{t_i} \bigr] \dd{s} \notag \\
          &\quad - \E \int_{t_i}^{t_{i+1}} D_1 \nu(\bar{\bZ}_{t_i}, s)
            \mathcal{S}(\bY_{t_i})
            \E\bigl[\DeltaSYSY W_{t_i,s}
            \given \mathcal{F}_{t_i} \bigr] \dd{s} \notag \\
          &\quad - \sum_{j=2}^{d} \E \int_{t_i}^{t_{i+1}}
            D_j \nu(\bar{\bZ}_{t_i}, s)
            \E\bigl[\DeltaSYSY Y^{j-1}_{t_i,s}
            \given \mathcal{F}_{t_i} \bigr] \dd{s} \notag\\
          &\quad -  \E \int_{t_i}^{t_{i+1}}  D_{11}\nu(\bar{\bZ}_{t_i}, s)
            \mathcal{S}(\bY_{t_i})^2
            \E\bigl[ \DeltaSYSY (W_{t_i,s})^2
            \given \mathcal{F}_{t_i} \bigr] \dd{s} \notag\\
          &\quad - 2 \sum_{j=2}^{d} \E \int_{t_i}^{t_{i+1}}
            D_{j1} \nu(\bar{\bZ}_{t_i}, s)  \mathcal{S}(\bY_{t_i}) 
            \E\bigl[ \DeltaSYSY W_{t_i,s}  Y^{j-1}_{t_i,s}
            \given \mathcal{F}_{t_i} \bigr] \dd{s}
            \notag\\
          &\quad -  \sum_{j,k=2}^{d}\E \int_{t_i}^{t_{i+1}}
            D_{jk} \nu(\bar{\bZ}_{t_i}, s) 
            \bigl[ \DeltaSYSY Y^{j-1}_{t_i,s}  Y^{k-1}_{t_i,s}
            \given \mathcal{F}_{t_i}\bigr] \dd{s} \notag\\
          &\quad - \E \int_{t_i}^{t_{i+1}} \E\bigl[ \DeltaSYSY \mathcal{R}_3(\nu)
            \given \mathcal{F}_{t_i} \bigr] \dd{s}  \notag \\
          &\eqdef J_{0} + J_{1,0} + J_{1,1} + J_{2,0} + J_{2,1} + J_{2,2}
            - \E \int_{t_i}^{t_{i+1}} \E\bigl[ \DeltaSYSY \mathcal{R}_3(\nu)
            \given \mathcal{F}_{t_i} \bigr] \dd{s} \,. \label{eq:J-splitting}
\end{align}
Here $J_{0}$ is as before and the $J_{k,i}$ involve $k$th order
derivatives of $\nu$ that do not necessarily vanish for general payoff
functions $\varphi$ (here the second index $i \leq k$ is the number of
the derivatives that correspond to extended variable directions and
hence the number of sums over extended variable indices). Analogously,
for $\widetilde{J}$, the term corresponding to the increment $\DeltaSY$,
we have,
\begin{equation}
   \label{eq:Jtilde-splitting}
  \widetilde{J} = \widetilde{J}_{0} + \widetilde{J}_{1,0} + \widetilde{J}_{1,1}
                + \widetilde{J}_{2,0} + \widetilde{J}_{2,1} + \widetilde{J}_{2,2}
                -  \E \int_{t_i}^{t_{i+1}} \E\bigl[ \DeltaSY \mathcal{R}_3(\widetilde{\nu})
                \given \mathcal{F}_{t_i} \bigr] \dd{s} \,,
\end{equation}
with $\widetilde{\nu}$ in place of $\nu$ and the increment $\DeltaSY$ in
place of $\DeltaSYSY$ compared to \cref{eq:J-splitting}. In the
sequel, we will simply write
\begin{equation}
  \label{eq:J-Jtilde-notation}
  J_{k} = \sum_{i = 0}^{k} J_{k,i} \qquad \text{and}
  \qquad \widetilde{J}_{k} = \sum_{i = 0}^{k} \widetilde{J}_{k,i}\,,\qquad k>0\,,
\end{equation}
for the sum of all terms involving $k$th order derivatives of $\nu$
and $\widetilde{\nu}$. In what follows, we first take the fBm view and
assume deterministic $\|D\nu\|_\infty = O(1)$ and similarly for
$\widetilde{\nu}$. Since the terms corresponding to $\nu$ and
$\widetilde{\nu}$ contribute only to the constant and not to the rate,
this assumption allows us to easily deduce the order in $\Delta t$,
namely, that terms $J_k$ are at least order $O (\Delta t^{H+3/2})$.
Expressing the $J_k$ in extended variables, as in
\cref{eq:J0-first-term-naive}, it is then possible to carrying out a
second round of Taylor expansions to demonstrates that the constants
are controlled.

\subsection{Estimate for general payoffs: $J_{0}$ is $O(\Delta t^{H+3/2})$}
\label{sec:estimate-J0}%

In \cref{sec:weak-rate-quad:one}, $\nu$ is deterministic and $O(1)$
since it depends on two derivatives of the quadratic payoff $\varphi$.
Continuing as in \cref{eq:J0-cond-expec},
our starting point for the full estimate for $J_{0}$ is
\begin{equation}
  \label{eq:J0}
  \begin{split}
    J_{0}%
    &= - c_H^2 \E \int_{t_i}^{t_{i+1}}\nu(\bar{X}_{t_i}, \bY_{t_i}\,;
    s) \sum_{k,l=1}^{N_L} \Delta\theta_k \Delta\theta_l \bigl[
    Y^k_{t_i} Y^l_{t_i} - \tfrac{1}{\theta^p_k+\theta^p_l}\bigr] g(s)
    \dd{s} \,,
  \end{split}
\end{equation}
where we emphasize the dependence of $\nu$ on $\bY_{t_i}$. We define
an auxiliary function,
\begin{equation}
  \label{eq:J0-f-kl}
  f^{kl}_s(Y^k_{t_i}, Y^l_{t_i})
  \defeq \E\bigl[ \nu(\bar{X}_{t_i}, \bY_{t_i}\,; s)
  \given Y^k_{t_i}, Y^l_{t_i} \bigr]\,,
\end{equation}
and expand $f^{kl}_s$ in a Taylor series at zero,
\begin{equation}
  \label{eq:taylor-f-kl}
  f^{kl}_s(Y^k_{t_i}, Y^l_{t_i}) = f^{kl}_s(\boldsymbol{0})
  + \sum_{\alpha \in \mathcal{J}_{\alpha_{\max}}} \frac{1}{|\alpha|!} \partial^\alpha_{kl}
  f^{kl}_s(\boldsymbol{0}) \bigl(Y^k_{t_i} Y^l_{t_i} \bigr)^{\alpha}
  + R_{\alpha_{\max}}(Y^k_{t_i},Y^l_{t_i})\,,
\end{equation}
for a set of multiindices
$\mathcal{J}_{\alpha_{\max}} = \{ \alpha = (\alpha_1, \alpha_2) \;:\;
1 \leq |\alpha| < \alpha_{\max} \}$ where we use the notation,
\begin{equation*}
  \partial^\alpha_{l_1 \dots l_j}
  = \partial_{l_1}^{\alpha_1}\dots \partial_{l_j}^{\alpha_j}\,,
  \qquad \alpha = (\alpha_1, \dots, \alpha_j)\,,
\end{equation*}
(as opposed to $D$) to emphasize that the derivatives are taken with
respect to $y_l$ directions only. The remainder is given by,
\begin{equation}
  \label{eq:aux-remainder}
  R_{\alpha_{\max}}(Y^k_{t_i}, Y^l_{t_i})
  \defeq \sum_{|\alpha| = \alpha_{\max}} \frac{1}{\alpha_{\max}!}
  \partial^\alpha_{kl} f^{kl}_s(\xi_k, \xi_l)
  (Y^k_{t_i} Y^l_{t_i})^{\alpha}\,,
\end{equation}
for an intermediate point $(\xi_{k}, \xi_{l})$.

Plugging this second round of Taylor expansions \cref{eq:taylor-f-kl}
into \cref{eq:J0}, yields
\begin{equation}
  \label{eq:J0-taylor-f}
  \begin{split}
    J_{0} &= -c_H^2 \Biggl( \sum_{k,l=1}^{N_L} \E\bigl[Y^k_{t_i}
    Y^l_{t_i} - \tfrac{1}{\theta^p_k+\theta^p_l}\bigr]
    \int_{t_i}^{t_{i+1}} f^{kl}_s(\boldsymbol{0}) g(s-t_i) \dd{s}
    \Delta\theta_k\Delta\theta_l\\
    &\quad+ \sum_{\alpha\in\mathcal{J}} \frac{1}{|\alpha|!}
    \sum_{k,l=1}^{N_L} \E\bigl[ (Y^k_{t_i})^{\alpha_1}
    (Y^l_{t_i})^{\alpha_2} \bigl(Y^k_{t_i}Y^l_{t_i} -
    \tfrac{1}{\theta^p_k+\theta^p_l}\bigr)\bigr] \int_{t_i}^{t_{i+1}}
    \partial^\alpha_{kl}f^{kl}_s(\boldsymbol{0}) g(s-t_i) \dd{s}
    \Delta\theta_k\Delta\theta_l\\
    &\quad+ \sum_{k,l=1}^{N_L} \E\bigl[ R_{\alpha_{\max}}
    (Y^k_{t_i},Y^l_{t_i}) \bigl( Y^k_{t_i}Y^l_{t_i} -
    \tfrac{1}{\theta^p_{k}+\theta^p_{l}} \bigr) \bigr] \Delta\theta_k
    \Delta\theta_l \int_{t_i}^{t_{i+1}}g(s-t_i)\dd{s} \Biggr)\,,
  \end{split}
\end{equation}
where $g(s)$ as in \cref{eq:lip-func-naive}. For the remainder term in
\cref{eq:J0-taylor-f}, we use the Hölder regularity of the fBm to
estimate the derivative of the auxiliary function evaluated at an
intermediate point,
\begin{equation*}
  \sum_{k,l=1}^{N_L}\E\bigl[ |R_{\alpha_{\max}}(Y^k_{t_i}, Y^l_{t_i})|
  |Y^k_{t_i}| | Y^l_{t_i}|\bigr] 
  \sim \E\bigl[|W_{t_i,t_{i+1}}^H|^{\alpha_{\max}}\bigr]
  \lesssim  \Delta t^{H \alpha_{\max}}\,.
\end{equation*}
From this last expression, the number of terms in the auxiliary
expansion, $\alpha_{\max}$, is finite and the contribution from the
remainder can be made to be order one by choosing
\begin{equation*}
  \alpha_{\max} \defeq \lceil \tfrac{1}{H}  \rceil \,.
\end{equation*}
For the remaining terms in \cref{eq:J0-taylor-f}, the $f^{kl}$ can be
estimated by the payoff $\varphi$ (see \cref{lem:bound-f-by-phi} in
\cref{sec:bound-f-varphi}) and therefore we write for convenience that
$f^{kl}$ and all derivatives are bounded by a constant $Q$,
\begin{equation}
  \label{eq:assume-f-J0}
  f^{kl} \in C^{\alpha_{\max}}_b
  \qquad \text{and} 
  \qquad \| D^\alpha f^{kl} \|_{\infty}\leq Q \quad \text{a.s.}
\end{equation}

To estimate the first term in \cref{eq:J0-taylor-f}, we use the
Lipschitz argument from the proof of weak rate one for quadratic
payoffs in \cref{eq:lip-func-naive} together with
\cref{eq:assume-f-J0}, and find that,
\begin{equation}
  \label{eq:J0-first-term-full-estimate}
  c_H^2 Q \Delta t^2 \sum_{k,l=1}^{N_L}
  |\E[Y^k_{t_i}Y^l_{t_i}(\theta^p_k+\theta^p_l) - 1] |
  \Delta\theta_k \Delta\theta_l 
  \leq \frac{c_H^2}{2p^2} Q \Gamma\big(\tfrac{1}{p})^2 t_i^{-2/p} \Delta t^2  \,,
\end{equation}
for a constant proportional to $t_i^{-1/p}$ as in
\cref{eq:est-J0-naive}. The key to obtaining the rate in $\Delta t$ for terms of higher order
in $\alpha$ again depends on demonstrating summability in $k$ and
$l$, as in \cref{eq:J0-first-term-full-estimate}. The higher order terms in \cref{eq:J0-taylor-f} are of the form,
\begin{equation*}
  \left\{\E \left[ (Y^k_{t_i})^{\alpha_1+1} (Y^l_{t_i})^{\alpha_2+1}\right] - \E \left[ (Y^k_{t_i})^{\alpha_1} (Y^l_{t_i})^{\alpha_2}\right] \tfrac{1}{\theta_k^p + \theta_l^p} \right\} \int_{0}^{\Delta t} \partial^\alpha_{kl} f_{s-t_i}(\boldsymbol{0}) g(s) \dd s\,,
\end{equation*}
and we use Isserlis' theorem to expand the expectations into products
of covariances of the extended variables. For $|\alpha|$ odd,
e.g.~when $|\alpha| = 1$ as $\alpha_1=1\,, \alpha_2=0$ or
$\alpha_1=0\,, \alpha_2=1$, then the term is zero by Isserlis'. For
$|\alpha|$ even, we use the exact expression for the covariance and
then check the summability in $k$ and $l$. In contrast to
\cref{eq:J0-first-term-full-estimate}, for these higher order terms we
obtain the rate $O(\Delta t^{H+3/2})$.

For example when $|\alpha|=2$, we apply Isserlis' theorem and obtain
a term containing,
\begin{equation}
  \label{eq:J0-second-term-alpha2-full-estimate}
  \begin{split}
    \E[ Y^k_{t_i}Y^k_{t_i}] \E[Y^l_{t_i} Y^l_{t_i}] & + 2 \E[Y^k_{t_i}Y^l_{t_i}]^2 - \E[Y^k_{t_i}Y^l_{t_i}] \tfrac{1}{\theta_k^p+\theta_l^p}\\
   &= \frac{1-e^{-2\theta_k^p t_i}}{2\theta_k^p} \frac{1-e^{-2\theta_l^p t_i}}{2\theta_l^p} + 2 \left(\frac{1-e^{-(\theta_k^p+\theta_l^p)t_i}}{\theta_k^p+\theta_l^p}\right)^2
        - \frac{1-e^{-(\theta_k^p+\theta_l^p) t_i}}{(\theta_k^p + \theta_l^p)^2}\,.
  \end{split}
\end{equation}
Using the Lipschitz argument as in
\cref{eq:J0-first-term-full-estimate}, the final two terms in
\cref{eq:J0-second-term-alpha2-full-estimate}, containing powers of
$(\theta^p_k + \theta^p_l)$ in the denominator, are
summable in $k$ and $l$ yielding the estimate $O(\Delta t^2)$.
Neglecting $\partial^\alpha f_s$, we focus on the contribution to $J_0$
from first term in \cref{eq:J0-second-term-alpha2-full-estimate},
\begin{equation*}
  F(\theta_k^p) F(\theta_l^p) \int_0^{\Delta t} (1-e^{-s(\theta^p_k + \theta^p_l)}) \dd s \Delta\theta_k \Delta\theta_l \leq \underbrace{F(\theta_k^p) F(\theta_l^p) \left(1 - e^{-(\theta_k^p+\theta_l^p)\Delta t} \right)}_{\eqdef G(\theta_k^p, \theta_l^p, \Delta t)}
  \Delta t \Delta \theta_k \Delta \theta_l
\end{equation*}
where we introduce the notation
\begin{equation*}
  F(u) \defeq \frac{1-\exp(-2u t_i)}{2u} \,.
\end{equation*}
We thus consider the sum 
\begin{equation}
  \label{eq:J0-G-sum}
  \sum_{(k,l) \in \mathbf{N}^2} G(\theta_k^p, \theta_l^p, \Delta t) \Delta t \Delta \theta_k \Delta \theta_l
\end{equation}
and obtain an upper bound $\tilde{C}\Delta t^{H+3/2}$, for a constant
independent of $\Delta t$, by partitioning according to the four cases
below; we let $C, C', C'' >0$ and $\alpha \in (0,1)$ be constants that
are independent of $\Delta t$.

For the first case, we consider 
\begin{equation*}
\mathcal{N}_{1} \defeq \{ (k,l) \in \mathbf{N}^2 : \theta_k^p+\theta_l^p \leq C\}\,.
\end{equation*}
Since $1 - \exp(-u) \leq \min(1, u)$, for $u >0$, the summand in \cref{eq:J0-G-sum} is
bounded by
\begin{equation*}
  G(\theta_k^p, \theta_l^p, \Delta t) \leq t_i^2 C \Delta t\,,
\end{equation*}
for $(k,l) \in \mathcal{N}_1$, and thus
\begin{equation*}
  \sum_{(k,l) \in \mathcal{N}_1} G(\theta_k^p, \theta_l^p, \Delta t) \Delta t \Delta \theta_k \Delta \theta_l = O(\Delta t^2)\,.
\end{equation*}
In the second case, we consider
\begin{equation*}
  \mathcal{N}_{2} \defeq \{ (k,l) \in \mathbf{N}^2 : C \leq \theta_k^p + \theta_k^p \leq C' \Delta t^{-\alpha}\}\,.
\end{equation*}
For $(k,l) \in \mathcal{N}_2$, we estimate
\begin{equation*}
  1 - e^{-(\theta_k^p + \theta_l^p) \Delta t} \leq 1 - e^{- C' \Delta t^{1-\alpha}} \leq C' \Delta t^{1-\alpha}
\end{equation*}
and also
\begin{equation*}
  F(\theta^p) = \frac{1-e^{-2\theta^p t_i}}{2\theta^p} \leq \min( t_i, \tfrac{1}{2\theta^p}) \eqdef m(\theta^p)\,,
\end{equation*}
so that the summand in \cref{eq:J0-G-sum} is bounded by
\begin{equation*}
  G(\theta_k^p, \theta_l^p, \Delta t) \leq C' \Delta t^{1-\alpha} m(\theta_k^p) m(\theta_l^p)\,.
\end{equation*}
Although $\mathcal{N}_2$ grows as $\Delta t \to 0$, the order one contribution from $m$ is overtaken by the decay in $\theta$, and thus,
\begin{equation*}
  \sum_{(k,l) \in \mathcal{N}_2} G(\theta_k^p, \theta_l^p, \Delta t) \Delta t \Delta \theta_k \Delta \theta_l  = O(\Delta t^{2-\alpha})\,.
\end{equation*}

In the third case, we let
\begin{equation*}
  \mathcal{N}_3 \defeq \{ (k,l) \in \mathbf{N}^2 : C' \Delta t^{-\alpha} \leq \theta_k^p + \theta_l^p \leq C'' \Delta t^{-1} \}\,;
\end{equation*}
this is the most critical case, where our estimate must be sharpest. In particular, first observe that the function $1 - \exp(-(\theta_k^p + \theta_l^p) \Delta t)$
varies from $O(\Delta t^{1-\alpha})$ for values $\theta_k^p + \theta_l^p = O(\Delta t^{-\alpha})$ up to $O(1)$ for values $\theta_k^p + \theta_l^p = O(\Delta t^{-1})$ and we exploit this variation to achieve our estimate. Recalling that $1 - \exp(-u) \leq u$, for $0\leq u \leq 1$, then for $(k,l) \in \mathcal{N}_3$ (taking $C' = 1$) we bound
\begin{equation*}
  1 - e^{-(\theta_k^p + \theta_l^p) \Delta t} \leq (\theta_k^p + \theta_l^p) \Delta t \,.
\end{equation*}
Thus,
\begin{equation*}
  \sum_{(k,l) \in \mathcal{N}_3} G(\theta_k^p, \theta_l^p, \Delta t) \Delta t \Delta\theta_k \Delta \theta_l \leq \sum_{(k,l) \in \mathcal{N}_3} m(\theta_k^p) m(\theta_l^p) (\theta_k^p + \theta_l^p) \Delta t^2 \Delta \theta_k \Delta \theta_l
\end{equation*}
so that we are left to estimate the quantity on the right-hand side which blows up with a certain rate on $\Delta t$ as $\Delta t \to 0$. We bound the sum by the corresponding integral, i.e., %
\begin{equation*}
    \sum_{(k,l) \in \mathcal{N}_3}
    m(\theta_k^p) m(\theta_l^p) (\theta_k^p + \theta_l^p) \Delta \theta_k \Delta \theta_l \leq \iint \limits_{\substack{\Delta t^{-\alpha} \leq \theta_k^p + \theta_l^p \leq \Delta t^{-1}\\ \theta_j \geq 0}} \frac{\theta_k^p+\theta_l^p}{\max(t_i^{-1}, \theta_k^p) \max(t_i^{-1}, \theta_l^p)} \dd \theta_k \dd \theta_l %
\end{equation*}
and further divide the right-hand side above into three integrals,
$I_1$, $I_2$, and $I_3$, according to the regions
\begin{align*}
   R_1 &\defeq \{ \Delta t^{-\alpha} \leq \theta_j^p \leq \Delta t^{-1}  \quad \text{for } j = k,l \}\,,\\
      R_2 &\defeq \{ \Delta t^{-\alpha} \leq \theta_k^p \leq \Delta t^{-1}\quad \text{and} \quad 0 \leq \theta_l^p \leq \Delta t^{-\alpha} \}\,,\\
     R_3&\defeq \{ \Delta t^{-\alpha} \leq \theta_k^p+\theta_l^p  \leq \Delta t^{-1} \quad \text{and} \quad  0 \leq \theta_j^p \leq \Delta t^{-\alpha} \quad \text{for } j = k,l \}\,,
\end{align*}
respectively. Then,
\begin{align*}
  I_1 &\leq \iint \limits_{R_1} \left(\theta_l^{-p} + \theta_k^{-p} \right) \dd \theta_k \dd \theta_l \\
      & \leq 2 \int \limits_{\Delta t^{-\alpha/p}}^{\Delta t^{-1/p}} \frac{\dd \theta_k}{\theta_k^p} = \left. - \frac{\theta_k^{-(p-1)}}{p-1} \right|_{\Delta t^{-\alpha/p}}^{\Delta t^{-1/p}}
        \lesssim \Delta t^{\alpha(1-1/p)}\,,
\end{align*}
and we note that $\Delta t^{\alpha(1-1/p)} \to 0$ as $\Delta t \to 0$, since $0 < \alpha < 1$ and $p>2$, so this is not the dominant term. For $R_2$,
\begin{align*}
  I_2 &\lesssim \int \limits_{\Delta t^{-\alpha/p}}^{\Delta t^{-1/p}} \int \limits_0^{\Delta t^{-\alpha/p}}
        \frac{(\theta_k^p + \theta_l^p) }{\max(t_i^{-1}, \theta_k^p) \max(t_i^{-1}, \theta_l^p)} \dd \theta_l \dd \theta_k \\
      &\lesssim \int \limits_{\Delta t^{-\alpha/p}}^{\Delta t^{-1/p}} \int \limits_{t_i^{-1/p}}^{\Delta t^{-\alpha/p}} \frac{\theta_k^p + \theta_l^p}{\theta_k^p\theta_l^p} \dd \theta_l \dd \theta_k + \int \limits_0^{t_i^{-1/p}} \int \limits_{\Delta t^{-\alpha/p}}^{\Delta t^{-1/p}} \frac{\theta_k^p+\theta_l^p}{\theta_k^p t_i^{-1}} \dd \theta_k \dd \theta_l \,.
\end{align*}
The first integral is equal to 
\begin{equation*}
  \int \limits_{\Delta t^{-\alpha/p}}^{\Delta t^{-1/p}} \int \limits_{t_i^{-1/p}}^{\Delta t^{-\alpha/p}} (\theta_l^{-p}+\theta_k^{-p}) \dd \theta_k \dd \theta_l = 2 \Delta t^{-\alpha/p} \int \limits_{\Delta t^{-\alpha/p}}^{\Delta t^{-1/p}} \theta^{-p}_k \dd \theta_k = \Delta t^{-\alpha/p}
  \left. \frac{\theta^{-p+1}}{-p+1} \right|_{\Delta t^{-\alpha/p}}^{\Delta t^{-1/p}} %
  = \Delta t^{\alpha(1-2/p)}\,,
\end{equation*}
which we note converges to $0$ as $\Delta t \to 0$ since $p > 2$, $0 < \alpha<1$ and therefore this is not the dominant term. On the other hand, the second integral is equal to
\begin{equation*}
  \int \limits_0^{t_i^{-1/p}} \int \limits_{\Delta t^{-\alpha/p}}^{\Delta t^{-1/p}} \left( \frac{1}{t_i^{-1}} + \frac{\theta_l^p}{\theta_k^p t_i^{-1}} \right) \dd \theta_k \dd \theta_l = O(\Delta t^{-1/p}) + O(1) \left( \int \limits_{\Delta t^{-\alpha/p}}^{\Delta t^{-1/p}} \theta^{-p}_k \dd \theta_k \right)\,,
\end{equation*}
which diverges as $\Delta t \to 0$ and therefore this is the dominant term. For $R_3$, 
\begin{equation*}
  I_3 \lesssim \int \limits_0^{\Delta t^{-\alpha/p}} \int \limits_0^{\Delta t^{-\alpha/p}}  \frac{(\theta_k^p + \theta_l^p) }{\max(t_i^{-1}, \theta_k^p) \max(t_i^{-1}, \theta_l^p)} \dd \theta_l \dd \theta_k  \,,
\end{equation*}
can be seen to converge to zero as $\Delta t \to 0$, following similar
reasoning to $I_2$, and is therefore not the dominant term. In
conclusion, we get that the contribution from
$(k,l) \in \mathcal{N}_3$ (after summing over $t_i$) is
\begin{equation*}
  \Delta t^{1-1/p} = \Delta t^{H+1/2}\,,
\end{equation*}
and it seems that we can take $\alpha$ as close to zero as we want.

In the fourth and final case, we consider
\begin{equation*}
  \mathcal{N}_4 = \{ (k,l) \in N^2 : C'' \Delta t^{-1} \leq \theta_k^p + \theta_l^p\}\,,
\end{equation*}
where we bound the summand in \cref{eq:J0-G-sum},
\begin{equation*}
  G(\theta_k^p, \theta_l^p, \Delta t) \leq m(\theta_k^p) m(\theta_l^p)\,,
\end{equation*}
for $(k,l) \in \mathcal{N}_4$. Now estimating the sum by the corresponding integral yields,
\begin{equation*}
  \sum_{(k,l) \in \mathcal{N}_4} \frac{\Delta \theta_k \Delta \theta_l}{\max(t_i^{-1}, \theta_k^p) \max(t_i^{-1}, \theta_l^p)} \lesssim \iint \limits_{\substack{\theta_k^p + \theta_l^p \geq C''\Delta t^{-1}\\ \theta_j \geq 0}} \frac{\dd \theta_k \dd \theta_l}{\max(t_i^{-1}, \theta_k^{p}) \max(t_i^{-1}, \theta_l^{p})} \,,
\end{equation*}
and we again consider three integrals, $I_1$, $I_2$ and $I_3$, now corresponding to the regions
\begin{align*}
  R'_1 &\defeq \{ \theta_j^p \geq \Delta t^{-1} \quad \text{for } j=k,l\}\,,\\
  R'_2 &\defeq \{ \theta_k^p \geq \Delta t^{-1} \quad \text{and} \quad \theta_l^p \leq \Delta t^{-1}\}\,,\\
  R'_3 &\defeq \{\theta_k^p + \theta_l^p \geq \Delta t^{-1} \quad \text{and} \quad \theta_j^p \leq \Delta t^{-1} \quad \text{for } j=k,l\}\,,
\end{align*}
respectively. Then,
\begin{align*}
  I_1 &= \iint \limits_{R'_1} \theta_{k}^{-p} \theta_l^{-p}\dd \theta_k \dd \theta_l
  = \left(\;\int \limits_{\Delta t^{-1/p}}^{\infty} \theta_k^{-p} \dd \theta_k \right)^2 = \left( - \left. \frac{\theta^{-p+1}}{-p+1} \right|_{\Delta t^{-1/p}}^{\infty} \right)^2 = O( \Delta t^{2(1-1/p)})\,,
\end{align*}
\begin{align*}
  I_2 &= 2 \iint \limits_{R'_2} \frac{\dd \theta_k \dd \theta_l}{\theta_k^p \max(t_i^{-1}, \theta_l^{p})} =
       O(\Delta t^{1 - 1/p}) \underbrace{\left(t_i t_i^{-1/p} + \cdots\right)}_{O(1)} = O(\Delta t^{1 - 1/p})\,,
\end{align*}
and
\begin{align*}
  I_3 &\simeq \iint \limits_{R'_3} \frac{\dd \theta_k \dd \theta_l}{\max(\theta_k^p, t_i^{-1}) \max(\theta_l^p, t_i^{-1})} 
  \lesssim \underbrace{ \int \limits_0^{\Delta t^{-1/p}} \frac{\dd \theta_l}{\max(\theta_l^p, t_i^{-1})}}_{O(1)} \;
  \underbrace{ \int \limits_{\frac{1}{2} \Delta t^{-1/p}}^{\Delta t^{-1/p}} \frac{\dd \theta_k}{\max(\theta_k^p, t_i^{-1})}}_{O(\Delta t^{(p-1)/p})} = O(\Delta t^{1-1/p})\,.
\end{align*}
Then taking the three integrals into account,
\begin{equation*}
  \sum_{(k,l) \in \mathcal{N}_4} G(\theta_k^p, \theta_l^p, \Delta t) \Delta t \Delta \theta_k \Delta \theta_l = O(\Delta t^{1-1/p}) \Delta t\,,
\end{equation*}
and we observe that the estimate in this region does not depend on $\alpha$ and therefore, recalling $p = 2/(1-2H)$, we obtain $\Delta t^{H+1/2}$ as the rate (after summing over $t_i$).

Altogether, our
full estimate for $J_{0}$ is then
\begin{equation}
  \label{eq:est-J0-full}
  |J_{0}| \lesssim C(H, Q, \alpha_{\max}, t_i) \Delta t^{H+3/2} 
  \lesssim O(\Delta t^{H+3/2}) \,,
\end{equation}
for a constant $C$ independent of $\Delta t$. In the next section, we consider estimates for the terms
$J_{1} $ and $\widetilde{J}_{1}$ in
\cref{eq:J-splitting,eq:Jtilde-splitting}. Inspired by the observation
that we obtain the same rate in \cref{eq:J0-first-term-full-estimate} as in
\cref{eq:est-J0-naive} with only the constant changing, we first work
directly with the fBm view to easily ascertain the rate. However, to
obtain the full estimate, we proceed as in this section by utilizing
the structure of the affine approximation to make a second round of
Taylor expansions and subsequently observing that the coefficients are
controlled.

\subsection{Estimate for general payoffs: $J_{1} $ and $\widetilde{J}_{1}$ are also
  $O(\Delta t^{H+3/2})$}

For the term $J_{1,0}$ in \cref{eq:J-splitting},
\begin{equation*}
  J_{1,0} = - \E \int_{t_i}^{t_{i+1}} D_1 \nu(\bar{X}_{t_i}, \bY_{t_i}; s) \mathcal{S}(\bY_{t_i}) \E\bigl[ \DeltaSYSY W_{t_i,s} \given \mathcal{F}_{t_i} \bigr]\,,
\end{equation*}
we first determine the order in $\Delta t$ which arises from the
conditional expectation. Working directly with the increment,
\begin{equation*}
  \DeltaSYSY = \widehat{W}^H_{t_i, s}(\widehat{W}^{H}_s + \widehat{W}^{H}_{t_i})
  = \widehat{W}^{H}_{t_i,s} (\widehat{W}^{H}_{t_i,s} + 2\widehat{W}^{H}_{t_i})\,,
\end{equation*}
we take the fBm view and express
\begin{equation*}
  \widehat{W}^{H}_{t_i,s} \sim W^{H}_{t_i,s}
  = \underbrace{\int_{0}^{t_i} K(s-r)
    - K(t_i-r) \dd{W_r}}_{\defeq V^H_{t_i}(s)}
  + \int_{t_i}^{s} K(s-r) \dd{W}_r \,,
\end{equation*}
using the power law kernel, i.e.~$K(r) = \sqrt{2H} r^{H-1/2}$, (inverse
discrete Laplace transform). We observe
\begin{equation*}
  \E[( W^H_{t_i,s})^2 W_{t_i,s} \given
  \mathcal{F}_{t_i}] = 2 V^H_{t_i}(s) \int_{t_i}^{s} K(s-r) \dd{r}\,, 
\end{equation*}
\begin{equation*}
  \E[W^H_{t_i,s} W_{t_i,s} \given
  \mathcal{F}_{t_i}] = \int_{t_i}^{s} K(s-r) \dd{r}\,, 
\end{equation*}
and
\begin{equation*}
  \bigl| \E[W^H_{t_i} V^H_{t_i}(s)] \bigr|
  = \left| \int_{0}^{t_i}K(t_i - r)K(s-r) - K^2(t_i-r) \dd{r} \right|
  \leq \left| \int_{0}^{t_i} K^2(t_i-r)\dd{r} \right| \sim t_i^{2H} \leq 1\,. 
\end{equation*}
Then,
\begin{align}
  J_{1,0} &= - \E \int_{t_i}^{t_{i+1}}  D_1 \nu(\bar{X}_{t_i}, \bY_{t_i}; s)
            \widehat{W}^H_{t_i} \E\bigl[
            \widehat{W}^{H}_{t_i,s} (\widehat{W}^{H}_{t_i,s}
            + 2\widehat{W}^{H}_{t_i}) W_{t_i,s} \given \mathcal{F}_{t_i} \bigr]
            \dd{s} \notag\\ %
          &\cong - \E \int_{t_i}^{t_{i+1}}  D_1 \nu(\bar{X}_{t_i}, \bY_{t_i}; s)
            W^H_{t_i}
            \bigl(\E[(W^H_{t_i,s})^2 W_{t_i,s} \given
            \mathcal{F}_{t_i}] + 2 W^H_{t_i}
            \E[W^H_{t_i,s} W_{t_i,s} \given
            \mathcal{F}_{t_i}]\bigr)\dd{s} \notag\\
          &\cong - 2 \E \int_{t_i}^{t_{i+1}}  D_1 \nu(\bar{X}_{t_i}, \bY_{t_i}; s)
            W^H_{t_i} (V^H_{t_i}(s)+W^H_{t_i})
            \underbrace{\int_{t_i}^{s}K(s-r) \dd{r}}_{\propto (s-t_i)^{H+1/2}}
            \dd{s}\,, \label{eq:J10-naive}
\end{align}
which suggests $J_{1,0} = O(\Delta t^{H+3/2})$ provided the
coefficient is controlled. The order of $J_{1,0}$ in $\Delta t$ is
asymptotically exact in \cref{eq:J10-naive}, i.e.~there is only
estimation of the constant.

In general, $D_1\nu$ is random and depends on $\bY_{t_i}$. This
changes the coefficient in the estimate but not the order in
$\Delta t$, as was observed by comparing \cref{eq:est-J0-full} to
\cref{eq:est-J0-naive}. Following \cref{sec:estimate-J0}, we proceed
to estimate the coefficient by making a second round of Taylor
expansions by introducing an auxiliary function $f$. In this spirit,
we start directly from \cref{eq:J10-naive}, noting that
\begin{align*}
  V^H_{t_i}(s) \sim \widehat{V}^H_{t_i}(s)
  \defeq c_H \sum_{l=1}^{N_L} Y^l_{t_i}
  (e^{-\theta^p_l(s-t_i)}-1) \Delta\theta_l  \,,
\end{align*}
since
\begin{align*}
  \int_{0}^{t_i} K(s-r) \dd{W_r}
  &= \int_{0}^{t_i} \sqrt{2H} (s-r)^{H-1/2} \dd{W_r} \\
  &= \int_{0}^{t_i}\frac{\sqrt{2H}}{\Gamma(\tfrac{1}{2}-H)}
    \int_{0}^{\infty}\theta^{-(H+1/2)} e^{-\theta(s-r)}\dd{\theta}\dd{W_r}\\
  &= \widetilde{c}_H \int_{0}^{\infty}\theta^{-(H+1/2)}
    e^{-\theta(s-t_i)} \int_{0}^{t_i}
    e^{-\theta(t_i -r)} \dd{W_r} \dd{\theta} \\
  &= \widetilde{c}_H \int_{0}^{\infty}\theta^{-(H+1/2)}
    e^{-\theta(s-t_i)} \widetilde{Y}_{t_i} \dd{\theta}\\
  &\approx c_H \sum_{l=1}^{N_L} Y^l_{t_i}
    e^{-\theta^p_l(s-t_i)} \Delta\theta_l\,.
\end{align*}
Rewriting $J_{1,0}$ directly in terms of the extended variables, we
obtain
\begin{align*}
  J_{1,0}
  &= - 2 \E \int_{t_i}^{t_{i+1}} D_1 \nu(\bar{X}_{t_i}, \bY_{t_i}\,; s)
    \widehat{W}^H_{t_i} (\widehat{V}^H_{t_i} + \widehat{W}^H_{t_i}) g(s) \dd{s}
    \notag\\
  &= -2c_H^2 \sum_{k,l=1}^{N_L} \int_{t_i}^{t_{i+1}} \E\bigl[ D_1
    \nu(\bar{X}_{t_i}, \bY_{t_i}\,; s) Y^k_{t_i} Y^l_{t_i} \bigr]
    e^{-\theta_l^p(s-t_i)} g(s) \dd{s} \Delta \theta_k \Delta \theta_l \,,
    \label{eq:J10-evars}
\end{align*}
where
\begin{equation*}
  g(s) \defeq \int_{t_i}^{s}K(s-r)\dd{r}  \propto (s-t_i)^{H+1/2} \,.
\end{equation*}

Recalling that
$\nu(\bar{\bZ}_{t_i}, s) = \nu(\bar{X}_{t_i}, \bY_{t_i}\,; s) $ is a
deterministic function of $\bar{\bZ}_{t_i}$, we further write
\begin{equation*}
  \E[ D_1\nu(\bar{X}_{t_i}, \bY_{t_i}\,; s) Y^k_{t_i} Y^l_{t_i}]
  = \E\bigl[ \E[ D_1\nu(\bar{X}_{t_i}, \bY_{t_i}\,; s)
  \given Y^k_{t_i}, Y^l_{t_i}]
  Y^k_{t_i}, Y^l_{t_i} \bigr]\,,
\end{equation*}
where we define a new auxiliary function (here $\varphi$ is already
higher order compared to \cref{eq:J0-f-kl})
\begin{equation*}
  \label{eq:J10-f-kl}
  f^{kl}_s(Y^k_{t_i}, Y^l_{t_i})
  \defeq \E\bigl[ D_1 \nu(\bar{X}_{t_i}, \bY_{t_i}\,; s)
  \given Y^k_{t_i}, Y^l_{t_i} \bigr]\,.
\end{equation*}
We Taylor expand $f^{kl}$ at zero, yielding
\begin{equation*}
  \label{eq:J10}
  \begin{split}
    J_{1,0} &= -2 c_H^2\Biggl( \sum_{k,l=1}^{N_L} \E[Y^k_{t_i} Y^l_{t_i}]  \int_{t_i}^{t_{i+1}} f^{kl}_s(\boldsymbol{0}) e^{-\theta^p_l(s-t_i)}g(s) \dd{s} \Delta\theta_k \Delta\theta_l\\
    &\quad + \sum_{\alpha \in \mathcal{J}_{\alpha_{\max}}} \frac{1}{|\alpha|!}
    \sum_{k,l=1}^{N_L} \E\bigl[ (Y^k_{t_i})^{\alpha_1+1}
    (Y^l_{t_i})^{\alpha_2+1} \bigr]
    \int_{t_i}^{t_{i+1}} \partial^\alpha_{kl}f^{kl}_s(\boldsymbol{0}) e^{-\theta^p_l(s-t_i)}g(s)\dd{s} \Delta\theta_k \Delta\theta_l\\
    &\quad + \sum_{k,l=1}^{N_L}
    \E\bigl[R_{\alpha_{\max}}(Y^k_{t_i},Y^l_{t_i})Y^k_{t_i}Y^l_{t_i}\bigr]\int_{t_i}^{t_{i+1}}e^{-\theta^p_l(s-t_i)}g(s)
    \dd{s} \Delta\theta_k \Delta\theta_l \Biggr)
  \end{split}
\end{equation*}
where
$\mathcal{J}_{\alpha_{\max}} = \{\alpha=(\alpha_1, \alpha_2) \,:\, 1 \leq |\alpha| <
\alpha_{\max}\}$ and the remainder $R_{\alpha_{\max}}$ is given in the Lagrange form as in \cref{eq:aux-remainder}. Since $f^{kl}$ is bounded by $\varphi$ (see
\cref{lem:bound-f-by-phi} in \cref{sec:bound-f-varphi}), we write the
bound in terms of $Q$, as in \cref{eq:assume-f-J0}, to obtain
\begin{equation*}
  \label{eq:est-J10-full}
  |J_{1,0}| \lesssim C(H, Q, \alpha_{\max}) t_i^{1-2/p} \Delta t^{H+3/2} 
  \lesssim O(\Delta t^{H+3/2}) \,,
\end{equation*}
where we use Isserlis' theorem and the explicit form of the
covariances, e.g.,
\begin{align*}
  \sum_{k,l=1}^{N_L}\E[Y^k_{t_i}Y^l_{t_i}] \Delta\theta_k\Delta\theta_l
  &= \sum_{k,l=1}^{N_L} \frac{\Delta\theta_k \Delta\theta_l}{\theta^p_k+\theta^p_l}
    (1-e^{-(\theta^p_k+\theta^p_l)t_i})\\
  &\leq \frac{2\pi}{4} \int_{0}^{\infty} \theta^{1-p} (1-e^{-\theta^p t_i}) \dd{\theta}\\
  &\leq \frac{2\pi}{4} \frac{1}{p-1} \Gamma(\tfrac{2}{p}) t_i^{1-2/p}\,,
\end{align*}
to observe that the sums over $\theta_k$ and $\theta_l$ converge
independently of $L, N_L$.

The term $J_{1,1}$ includes increments of OU processes arising from the
extended variables. Taking the fBm view using the power law kernel,
\begin{subequations}
  \begin{align}
    J_{1,1}
    &= - \sum_{j=1}^{N_L} \E \int_{t_i}^{t_{i+1}}
      D_{j+1} \nu(\bar{X}_{t_i}, \bY_{t_i}; s)
      \E\bigl[ \widehat{W}^H_{t_i,s} (\widehat{W}^H_{t_i,s}
      + 2 \widehat{W}^H_{t_i}) Y^j_{t_i,s} \
      \given \mathcal{F}_{t_i} \bigr] \dd{s} \notag \\
    &\cong - \sum_{j=1}^{N_L} \E \int_{t_i}^{t_{i+1}}
      D_{j+1} \nu(\bar{X}_{t_i}, \bY_{t_i}; s)
      \bigl((V^H_{t_i}(s))^2 + 2 W^H_{t_i} V^H_{t_i}(s)\bigr)
      Y^j_{t_i} (e^{-\theta^p_j(s-t_i)}-1)
      \dd{s} \label{eq:J11-naive:rate2}\\
    & \quad - \sum_{j=1}^{N_L}\E \int_{t_i}^{t_{i+1}}
      D_{j+1} \nu(\bar{X}_{t_i}, \bY_{t_i}; s)
      (2V^H_{t_i}(s) + 2W^H_{t_i})
      \int_{t_i}^{s}e^{-\theta^p_j (s-r)}K(s-r) \dd{r} \dd{s} 
      \label{eq:J11-naive:rateHp3half}\\
    & \quad - \sum_{j=1}^{N_L}\E \int_{t_i}^{t_{i+1}}
      D_{j+1} \nu(\bar{X}_{t_i}, \bY_{t_i}; s)
      Y^j_{t_i}(e^{-\theta^p_j(s-t_i)}-1) \int_{t_i}^{s}K^2(s-r)\dd{r}
      \dd{s} \,, \label{eq:J11-naive:rate2Hp2}
  \end{align}
\end{subequations}
since
\begin{equation*}
  \begin{split}
    \E[(W^H_{t_i,s})^2 Y^j_{t_i,s}\given \mathcal{F}_{t_i}] =
    Y^{j}_{t_i} (e^{-\theta^p_j(s-t_i)}-1)
    \int_{t_i}^{s}K^2(s-r)\dd{r}
    + 2V^H_{t_i}(s) \int_{t_i}^{s} K(s-r) e^{-\theta^p_j(s-r)}\dd{r}\\
    + (V^H_{t_i}(s))^2Y^j_{t_i} (e^{-\theta^p_j(s-t_i)}-1)
  \end{split}
\end{equation*}
and
\begin{equation*}
  \E[ W^H_{t_i,s} Y^j_{t_i,s} \given
  \mathcal{F}_{t_i}] =
  \int_{t_i}^{s}e^{-\theta^p_j(s-r)}
  K(s-r)\dd{r} +  V^H_{t_i}(s) Y^j_{t_i}
  (e^{-\theta^p_j(s-t_i)}-1)\,.
\end{equation*}
We examine the contributions to the rate in $\Delta t$ from each of
the terms
\cref{eq:J11-naive:rate2,eq:J11-naive:rateHp3half,eq:J11-naive:rate2Hp2};
after integrating over $s$ we find that \cref{eq:J11-naive:rate2}
yields $O(\Delta t^2)$, \cref{eq:J11-naive:rateHp3half} yields
$O(\Delta t^{H+3/2})$, and \cref{eq:J11-naive:rate2Hp2} yields
$O(\Delta t^{2H+2})$. Since
\cref{eq:J11-naive:rate2,eq:J11-naive:rate2Hp2} are higher order in
$\Delta t$ we examine only \cref{eq:J11-naive:rateHp3half}, where we
note
\begin{equation}
  \label{eq:J11-rateHp3half-key-term-est}
\begin{split}
  \int_{t_i}^{s}e^{\theta^p_j(s-t_i)}K(s-r)\dd{r}
  &= e^{-\theta^p_j s}\sqrt{2H} \int_{t_i}^{s}
    e^{\theta^p_j r} (s-r)^{H-1/2} \dd{r}\\
  &= \sqrt{2H} \theta^{-(H+1/2)p}_j \Bigl[
    \Gamma\bigl(H+\tfrac{1}{2}, (s-r)\theta^p_j\bigr)
    \Bigr]_{r=t_i}^{r=s}\\
  &= \sqrt{2H} \theta^{-(H+1/2)p}_j \Bigl[ e^{s-r}
    (s-r)^{H+1/2} \Gamma\bigl( H+1/2, \theta^p_j \bigr)
    \Bigr]_{r=t_i}^{r=s}\\
  &= - \sqrt{2H} \theta^{-(H+1/2)p}_j \Gamma\bigl( H+\tfrac{1}{2}, \theta^p_j \bigr)
    e^{s-t_i} (s-t_i)^{H+1/2}\
  \end{split}
\end{equation}
by a change of variable in the argument of the incomplete gamma
function where
\begin{equation*}
  q \defeq (H+1/2)p = \frac{2H+1}{1-2H} > 1\,, \quad H\in(0,1/2)\,.
\end{equation*}
This suggests $J_{1,1} = O(\Delta t^{H+3/2})$ since
the sum over $\theta_j$ converges (also helpful to note
$\Gamma\bigl( H+\tfrac{1}{2}, \theta^p_j \bigr) \to 0$ as
$\theta_j \to \infty$).

For the full estimate of $J_{1,1}$, we note
\begin{equation*}
  D_{j+1} \nu(\bar{X}_{t_i} \bY_{t_i}\,; s)
  = c_H \Delta \theta_j \E\bigl[ \varphi^{(3)}(\widehat{X}_T) M^j_{s,T}
  \given (\widehat{X}_s,\bY_s) = (\bar{X}_{t_i}, \bY_{t_i}) \bigr]\,,
\end{equation*}
where we assume that $M^j_{s,T} < \infty$ (see
\cref{eq:increment-value-func-condexp} for definition of $M$).
Omitting the higher order terms in $\Delta t$ in
\cref{eq:J11-naive:rate2,eq:J11-naive:rate2Hp2}, we return directly to
\cref{eq:J11-naive:rateHp3half},
\begin{equation*}
  J_{1,1} = - 2 c_H  \sum_{j=1}^{N_L} \E \int_{t_i}^{t_{i+1}}
  \E\bigl[ \varphi^{(3)}(\widehat{X}_T) M^j_{s,T}
  \given (\widehat{X}_s,\bY_s) = (\bar{X}_{t_i}, \bY_{t_i}) \bigr]
  (\widehat{V}^H_{t_i}(s)
  + \widehat{W}^H_{t_i}) g(s) h(\theta_j) \dd{s}  \Delta \theta_j \,,
\end{equation*}
where we let
\begin{equation*}
  g(s) \defeq - \sqrt{2H} e^{s-t_i} (s-t_i)^{H+1/2} 
\end{equation*}
and
\begin{equation*}
  h(\theta_j) \defeq \theta^{-q}_j
  \Gamma \bigl( H+\tfrac{1}{2}, \theta^p_j\bigr) \,.
\end{equation*}
We define a new auxiliary function
\begin{equation*}
  \label{eq:J11-f-k}
  f^k_s(Y^k_{t_i}) \defeq \E\bigl[
  \E[\varphi^{(3)}(\widehat{X}_T) M^j_{s,T} \given
  (\widehat{X}_s,\bY_s) = (\bar{X}_{t_i}, \bY_{t_i})]
  \given Y^k_{t_i} \bigr]\,, 
\end{equation*}
where, although the inner expectation depends implicitly on $j$, here
the index $k$ refers to the general component $Y^k_{t_i}$ that we are
conditioning against. Expanding $f^k$ at zero, we find
\begin{equation*}
  \begin{split}
    J_{1,1} &= - 2 c_H^2 \sum_{j,k=1}^{N_L} h(\theta_j)
    \int_{t_i}^{t_{i+1}} \E\bigl[ f^k_s(Y^k_{t_i}) Y^k_{t_i} \bigr]
    e^{-\theta^p_k (s-t_i)}
    g(s) \dd{s} \Delta\theta_j \Delta\theta_k\\
    &= - 2 c_H^2 \sum_{\alpha_1 = 1}^{\alpha_{\max} - 1}
    \frac{1}{\alpha_1 !} \sum_{j,k=1}^{N_L} h(\theta_j) \E\bigl[
    (Y^k_{t_i})^{\alpha_1+1} \bigr] e^{-\theta^p_k(s-t_i)}
    \int_{t_i}^{t_{i+1}}\partial^{\alpha_1}_k
    f^k_{s}(\boldsymbol{0}) g(s) \dd{s} \Delta\theta_j \Delta\theta_k\\
    &\quad - 2c_H^2 \sum_{j,k=1}^{N_L} h(\theta_j) \E\bigl[
    R_{\alpha_{\max}}(Y^k_{t_i}) Y^k_{t_i} \bigr]
    e^{-\theta^p_k(s-t_i)} \int_{t_i}^{t_{i+1}} g(s)\dd{s}
    \Delta\theta_j \Delta\theta_k \,.
  \end{split}
\end{equation*}
Following from the boundedness of $f$ and its derivatives (as in
\cref{eq:assume-f-J0}), we obtain the full estimate for $J_{1,1}$,
\begin{equation*}
  \label{eq:est-J11-full}
  |J_{1,1}| \lesssim C(H, Q, \alpha_{\max})  t_i^{1-2/p}  \Delta t^{H+3/2}
  \lesssim O(\Delta t^{H+3/2})\,,
\end{equation*}
where we use Isserlis' theorem and the representation of the
covariance for the extended variables to determine the coefficient
depending on $t_i$ (which we note is also summable over $i$).

The estimation of $\widetilde{J}_{1,0}$ and $\widetilde{J}_{1,1}$ follows the
program above. For $\widetilde{J}_{1,0}$ we begin by taking the fBm view
with the kernel $K$,
\begin{equation*}
  \begin{split}
    |\widetilde{J}_{1,0}| &= \Bigl| \E \int_{t_i}^{t_{i+1}}
    D_1\widetilde{\nu}(\bar{X}_{t_i}, \bY_{t_i} \,; s) W^H_{t_i}
    \E\bigl[ W^H_{t_i,s} W_{t_i,s} \given \mathcal{F}_{t_i} \bigr]\dd{s} \Bigr|\\
    &= c_H \Bigl| \sum_{k=1}^{N_L}\int_{t_i}^{t_{i+1}}
    D_1\widetilde{\nu}(\bar{X}_{t_i}, \bY_{t_i} \,; s) W^H_{t_i}
    \int_{t_i}^{s}\underbrace{K(s-r) \dd{r}}_{(s-t_i)^{H+1/2}} \dd{s} \Bigr|\\
  \end{split}
\end{equation*}
which suggests $\widetilde{J}_{1,0} = O(\Delta t^{H+3/2})$. We recall that
\begin{equation*}
  D_1 \widetilde{\nu}(\bar{\bZ}_{t_i}, s)
  = \E \bigl[ \varphi^{(3)}(\widehat{X}_T)
  \textstyle \bigl( c_H \sum_l M^l_{s,T} \Delta\theta_l \bigr)
  \given \bZ_s = \bar{\bZ}_{t_i}\bigr]\,,
\end{equation*}
where we assume that
\begin{equation*}
  c_H \sum_{l=1}^{N_L} M^l_{s,T} \Delta\theta_l < \infty \,,
\end{equation*}
and then define an auxiliary function
\begin{equation*}
  f^k_s(Y^k_{t_i}) \defeq \E\Bigl[ \E \bigl[ \varphi^{(3)}(\widehat{X}_T)
  \textstyle \bigl( c_H \sum_l M^l_{s,T} \Delta\theta_l \bigr)
  \given (\widehat{X}_s, \bY_s) = (\bar{X}_{t_i}, \bY_{t_i}) \bigr]
  \given Y^k_{t_i}\Bigr]\,.
\end{equation*}
Finally, Taylor expanding $f^k_s$ at zero we encounter terms similar
to those estimated previously yielding the full estimate,
\begin{equation*}
  \label{eq:est-Jtilde10-full}
  |\widetilde{J}_{1,0}| \lesssim O(\Delta t^{H+3/2})\,.
\end{equation*}

Likewise for $\widetilde{J}_{1,1}$, taking the fBm view with the
conditional expectation term
$\E[ \widehat{W}^H_{t_i,s} Y^j_{t_i,s} \given \mathcal{F}_{t_i}]$ yields
\begin{equation*}
  |\widetilde{J}_{1,1}| = \Bigl| \sum_{j=1}^{N_L} \E \int_{t_i}^{t_{i+1}}
  D_{j+1} \widetilde{\nu} (\bar{X}_{t_i}, \bY_{t_i}; s)
  \int_{t_i}^{s}e^{-\theta^p_j(s-r)}K(s-r)\dd{r} \dd{s}  \Bigr| \,,
\end{equation*}
where the estimate \cref{eq:J11-rateHp3half-key-term-est} (encountered
in $J_{1,1}$) suggest the rate $J_{1,1} = O(\Delta t^{H+3/2})$. For
the full estimate, we recall that
\begin{equation*}
  D_{j+1} \widetilde{\nu}(\bar{\bZ}_{t_i}, s)
  = \E \bigl[ \varphi^{(3)} (\widehat{X}_T)
  \textstyle \bigl( c_H \sum_l M^l_{s,T} \Delta\theta_l \bigr)^2
  \given \bZ_s = \bar{\bZ}_{t_i}\bigr] 
\end{equation*}
and expand in a second round of Taylor expansions for an appropriate
auxiliary function thereby obtaining
\begin{equation*}
  \label{eq:est-Jtilde11-full}
  | \widetilde{J}_{1,1} | \lesssim O(\Delta t^{H+3/2})\,.
\end{equation*}

Taken together, the estimates
\cref{eq:est-J10-full,eq:est-J11-full,eq:est-Jtilde10-full,eq:est-Jtilde11-full}
imply that the terms corresponding to first order derivatives of $\nu$
and $\widetilde{\nu}$ in \cref{eq:J-splitting,eq:Jtilde-splitting},
respectively, yields
\begin{equation*}
  \label{eq:est-summary-J1}
  |J_1| + |\widetilde{J}_1| = O(\Delta t^{H+3/2})\,,
\end{equation*}
using the notation in \cref{eq:J-Jtilde-notation}.

\subsection{Estimate  for general payoffs: remaining terms are higher order}

Additional terms \cref{eq:J-splitting,eq:Jtilde-splitting} appearing
in the expansion are higher order than ${H+3/2}$. For example, We find
$J_{2,0} = O(\Delta t^{2+2H})$ which can be seen by once again taking
the fractional view,
\begin{align*}
  \E\bigl[\DeltaSYSY  (W_{t_i,s})^2 \given \mathcal{F}_{t_i}\bigr]
  &= \E\bigl[\widehat{W}^H_{t_i,s}(\widehat{W}^H_{t_i,s} + 2\widehat{W}^H_{t_i}) (W_{t_i,s})^2\bigr]\\
  &= \E\bigl[ (\widehat{W}^{H}_{t_i,s})^2 (W_{t_i,s})^2 \bigr]
    + 2 \widehat{W}^H_{t_i}
    \underbrace{\E\bigl[\widehat{W}^H_{t_i,s} (W_{t_i,s})^2 \bigr]}_{\text{Isserlis'} \implies 0}\\
  &= \E[(\widehat{W}^H_{t_i,s})^2] \E[(W_{t_i,s})^2]
    + 2 (\E[\widehat{W}^H_{t_i,s} W_{t_i,s}])^2 \\
  &= \Delta t^{2H}\Delta t -  2(\Delta t^{H+1/2})^2 \eqdef g(\Delta t) \,,
\end{align*}
where $g (s) \propto (s-t_i)^{2H+1}$ and then expanding in a second
round of Taylor expansions for a suitable auxiliary function and
checking the control of the coefficient with respect to $\theta$,
\begin{equation*}
  |J_{2,0}| \lesssim  \frac{1}{2} \Bigl| \E \Bigl[  (W^H_{t_i})^2 
  \int_{0}^{\Delta t} D_{11} \nu(\bar{\bZ}_{t_i}, s+t_i) s^{2H+1} \dd{s}\Bigr] \Bigr|
  = O(\Delta t^{2+2H})\,.
\end{equation*}
The term $\widetilde{J}_{2,0}$ vanishes,
\begin{equation*}
  |\widetilde{J}_{2,0}| = \Bigl| \frac{1}{2} \E \Bigl[ \int_{t_i}^{t_{i+1}}D_{11}\widetilde{\nu}(\bar{\bZ}_{t_i}, s) (\widehat{W}^H_{t_i})^2
  \underbrace{\E \bigl[ \widehat{W}^H_{t_i,s} (W_{t_i,s})^2
    \bigr]}_{\text{Isserlis'} \implies 0} \dd{s} \Bigr] \Bigr| = 0 \,,
\end{equation*}
by Isserlis' theorem (although one could argue that the integrand is
formally contributing rate $H+2$ here for $\widetilde{J}_{2,0}$). The
terms involving increments of the OU processes yield similar results
(constants obtained would be better behaved owing to the additional
decay).

\subsection{Closure argument to finish proof of
  \cref{thm:weak-rate-gen}}

Thus far we have estimated the first few terms
\cref{eq:J-splitting,eq:Jtilde-splitting} arising from the Taylor
expansion in powers of $\Delta t$. Returning to the telescoping sum
\cref{eq:telescoping-sum-value-fcn}, we summarize our estimates up to
order $\kappa$,
\begin{align}
  \Err(T, \Delta t)
  &= \E\bigl[ \varphi(\widehat{X}_T) - \varphi(\bar{X}_{T}) \bigr] \notag\\
  &= \sum_{i=0}^{n-1}\left( \sum_{k=0}^{\kappa - 1}\left(  J_{k} + \widetilde{J}_{k} \right) + \mathcal{R}_\kappa(\nu) + \mathcal{R}_{\kappa} (\widetilde{\nu}) \right)\notag\\
  &\leq \sum_{i=0}^{n-1}\left( C_1(t_i)
    \Delta t^{H+3/2} + C_0(t_i) \Delta t^2 + O(\Delta t^{H+2}) + \mathcal{R}_\kappa(\nu) + \mathcal{R}_{\kappa} (\widetilde{\nu}) 
    \right) \,, \label{eq:close-expan-general}
\end{align}
using the notation in \cref{eq:J-Jtilde-notation}. Here the constants
$C_k$, for $k=0,\dots,\kappa-1$, depending on $t_i$, $H$, and $Q$, are
the coefficients that appear in the estimates
\cref{eq:est-J0-full,eq:est-J10-full,eq:est-J11-full}, etc., for
$J_{\cdot}$ and likewise for $\widetilde{J}_{\cdot}$. Importantly, each of
these constants was obtained independently of $N_L$, $L$, that is,
independently of the choice of parameter $\theta$ arising in the
Markovian extended variable state space formulation. Moreover, each
constant was summable in $t_i$. Interchanging the order of summation
in \cref{eq:close-expan-general} we obtain,
\begin{equation}
  \label{eq:close-remainder}
  \begin{split}
  \Err(T,\Delta t,\varphi)
  &\lesssim C_1' \Delta t^{H+1/2} + C_0' \Delta t
  + O(\Delta t^{H+1}) \\
  &\quad - \sum_{i=1}^{n-1} \E \int_{t_i}^{t_{i+1}} \Bigl(  \E\bigl[ \DeltaSYSY
    \mathcal{R}_\kappa(\nu) \given \mathcal{F}_{t_i} \bigr]
    + \E\bigl[ \DeltaSY
    \mathcal{R}_\kappa(\widetilde{\nu}) \given \mathcal{F}_{t_i} \bigr] \Bigr)\dd{s}\,,
  \end{split}
\end{equation}
with new constants
\begin{equation*}
  C_j' = \sum_{i=1}^{n-1} C_j(t_i) \Delta t\,. 
\end{equation*}
 
The remainders have a integral form \cref{eq:taylor-nu-remainder}, and
thus the conditional expectations in \cref{eq:close-remainder} are,
\begin{equation*}
  \begin{split}
    \E\bigl[ \DeltaSYSY \mathcal{R}_\kappa(\nu) \given
    \mathcal{F}_{t_i} \bigr] &= \frac{1}{\kappa !}
    \sum_{|\beta|=\kappa} (\widehat{W}^H_{t_i})^{\beta_1} \E\left[
    (\widehat{W}^H_{t_i,s})^2 (W_{t_i,s})^{\beta_1}
    (\bY_{t_i,s})^{\hat{\beta}} \int_0^1 D^\beta \nu (\bxi_\tau, s)
    \dd{\tau} \right]
  \end{split}
\end{equation*}
and
\begin{equation*}
  \begin{split}
    \E\bigl[ \DeltaSY \mathcal{R}_\kappa(\widetilde{\nu}) \given
    \mathcal{F}_{t_i} \bigr] &= \frac{1}{\kappa!}
    \sum_{|\beta|=\kappa} (\widehat{W}^H_{t_i})^{\beta_1}
    \E\left[\widehat{W}^H_{t_i,s} (W_{t_i,s})^{\beta_1}
    (\bY_{t_i,s})^{\hat{\beta}} \int_0^1 D^\beta \widetilde{\nu}
    (\bxi_\tau, s) \dd{\tau} \right]\,,
  \end{split}
\end{equation*}
where
\begin{equation*}
  \E\bigl[ \DeltaSYSY \mathcal{R}_\kappa(\nu) \given
  \mathcal{F}_{t_i} \bigr] \sim (\widehat{W}_{t_i,s}^H)^\kappa
  \qquad \text{and} \qquad
  \E\bigl[ \DeltaSY \mathcal{R}_\kappa(\widetilde{\nu}) \given
  \mathcal{F}_{t_i} \bigr] \sim (\widehat{W}_{t_i,s}^H)^\kappa\,.
\end{equation*}
We recall that 
\begin{equation*}
  \E | W^H_{t_i, t_{i+1}}|^\gamma
  \lesssim \Delta t^{\gamma H}\,,
\end{equation*}
by the Hölder continuity of the sample paths. Then by applying
Cauchy--Schwarz, we obtain in \cref{eq:close-remainder} that the
remainder terms yield $O(\Delta t^{\kappa H})$. Thus, $\kappa$ can be
chosen such that $\kappa > \frac{1}{H}$ yields a large but finite
expansion for general payoff functions $\varphi$. Then the error is is
given by
\begin{equation*}
  \Err(T, \Delta t,\varphi) \lesssim  C_1' \Delta t^{H+1/2} + C_0' \Delta t 
  +  O(\Delta t^{H+1})\,,
\end{equation*}
with weak error rate $H+1/2$ where all the coefficients are
controlled.

\begin{remark}[Kernel]
  In the proof of the main result \cref{thm:weak-rate-gen} and of
  \cref{thm:weak-rate-quad}, the specific form of $K$ in
  \cref{eq:kernel-explicit,eq:fBm} is not relevant and the spirit of
  the proof follows with any relevant $L^2$ kernel where the
  integrability conditions need to be checked.
\end{remark}

\section{Conclusions and outlook}

Rough stochastic volatility models are increasingly popular for option
pricing in quantitative finance. On the one hand,
the %
rough stochastic volatility overcomes empirical challenges to deliver
predictions consistent with observed market data. On the other hand,
the non-Markovian nature of
the fractional Brownian motion (fBm) driver is an impediment to both
theory and numerics. Despite the widespread use of
discretization-based simulation methods for option pricing under the
rough Bergomi model and the rough Stein--Stein model, few works have
studied the weak convergence rates that underpin this practice.  

For the rough Stein--Stein model, which treats the volatility as a
linear function of the driving fractional Brownian motion, we prove
that the weak convergence of the Euler scheme depends on the Hurst
parameter $H$ of the fBm driver and is weak rate $H+1/2$ for general
payoff functions (see \cref{thm:weak-rate-gen,thr:main}). Strong
numerical evidence is provided to support our theory. Our proof relies
on Taylor expansions for an extended variable system that is derived
from an affine Markovian approximation of the fBm drive. Our novel
approach also yields insights into unexpected behavior (that we
suspect has contributed to the consternation among experts regarding
the weak rate, as remarked in the footnote in
\cref{sec:introduction}). In particular, the expansions (see
\cref{thm:weak-rate-quad}) easily explain the better weak rate $1$
obtained for quadratic payoffs (see \cref{lem:weak-rate-quad-simple}).
This last point leads us to conjecture that the rate of convergence
for payoff functions well approximated by quadratic polynomials, as
seen from the law of the solution, may be hard to distinguish from
rate $1$ as illustrated in \cref{fig:shifted-cubic}. As stated in
\cref{rem:assumptions-main-theorem}, we do not doubt that the proof of
\cref{thm:weak-rate-gen} can be extended to nonlinear rough volatility
models such as the rough Bergomi model and this is the subject of
ongoing work.

\section{Grant information}

This work was supported by the KAUST Office of Sponsored Research
(OSR) under Award No.~OSR-2019-CRG8-4033 and the Alexander von Humboldt
Foundation. R.~Tempone is a member of the KAUST SRI Center for
Uncertainty Quantification in Computational Science and Engineering. A
portion of this work was carried out while E.~Hall was a Postdoctoral
Research Scientist in the Chair of Mathematics for Uncertainty
Quantification at RWTH Aachen University. C.~Bayer gratefully acknowledges
support by the German Research Council (DFG) via the Research Unit FOR 2402.

\section{Acknowledgements}

We are grateful to Andreas Neuenkirch for pointing out the simpler
proof in the case of quadratic payoffs presented in
\cref{sec:simple-proof-rate-one}.

\section{Declaration of interest statement}%
The authors report there are no competing interests to declare.

\printbibliography

\appendix

\section{Details of numerical implementation}
\label{sec:numerics}

We compute the weak error (here we assume $T=1$)
\begin{equation*}
  \bigl| \E[\varphi(X^{ref}_T) - \varphi(X^{\Delta t}_T)]\bigr| \,, 
\end{equation*}
for various payoff functions, $\varphi$, by the left-point scheme 
\begin{equation}
  \label{eq:numexp-forward-euler}
  \bar{X}_T(n) = \sum_{i=0}^{n} W^H_{t_{i}} (W_{t_{i+1}} - W_{t_i})\,,
\end{equation}
using a reference solution
$\bar{X}^{ref}_T \defeq \bar{X}_T(2^{12})$ and computations
$X^{\Delta t}_T = \bar{X}_T(n)$ for $\log_2n=\{6, \dots, 1\}$.
We sample paths $(W^H_{t_i}, W_{t_i})_{i\in[0 : n]}$ required for the
Monte Carlo approximation of \cref{eq:numexp-forward-euler} at points
of the reference mesh using the Cholesky decomposition method. For
each $\{H,T,n\}$ we first form the eigenvalue decomposition $(L,D)$ of
the full covariance matrix
\begin{equation*}
  \Sigma =
  \begin{pmatrix} \Sigma_{11} & \Sigma_{12}
    \\ \Sigma_{21} &\Sigma_{22}
  \end{pmatrix}\,,
\end{equation*}
with blocks $(\Sigma_{11})_{ij} = \cov(W^H_{t_i},W^H_{t_j})$,
$(\Sigma_{12})_{ij} = \cov(W^H_{t_i},W_{t_j})$,
$(\Sigma_{22})_{ij} = \cov(W_{t_i},W_{t_j})$. We refer to, e.g., Lemma
4.1 of \cite{BayerEtAl:2019st}, for the covariance function of the
Riemann-Liouville fBm and to \verb|pfq.m| from \cite{Huntley:2020rv}
to compute hypergeometric functions:
\begin{minted}[linenos=true,bgcolor=bg]{matlab}
t = T/n:T/n:T;
%
G = @(x) 2.0*H*( x.^(-gam)/(1.0-gam) + (gam*x.^(-1.0-gam)./(1.0-gam)) .* ...
    pfq([1.0, 1.0+gam], 3.0-gam, x.^(-1.0))/(2.0-gam) );
disp('Computing blocks S11 S12 S22')
[X,Y] = meshgrid(t,t);
Gmat = G((tril(Y./X)+tril(Y./X)') - eye(size(Y,1)).*diag(Y./X));
Gmat = Gmat.*~eye(size(Gmat)) + eye(size(Gmat)); %
S11 = ((tril(X)+tril(X)') - eye(size(X,1)).*diag(X)).^(2*H) .* Gmat;
S12 = sqrt(2*H)*(Y.^(H+0.5)-(Y-min(Y,X)).^(H+0.5))./(H+0.5);
S22 = min(X,Y);
disp('Finished blocks S11 S12 S22')
%
disp('Computing L D via eig')
[L,D] = eig([S11 S12; S12.' S22]);
\end{minted}
We then generate $M$ samples using the $LDL$ eigenvalue decomposition:
\begin{minted}[linenos=true,bgcolor=bg,escapeinside=!!]{matlab}
z = randn(2*n,M);
X = real(L*(D^0.5)*z);
WH = [zeros(1,M); X(1:n, :)];
W = [zeros(1,M); X(n+1:end, :)];
\end{minted}
and note that for $H=0.5$ the sample paths of $W^H$ generated by this
method is identical to $W$. One sample of $X^{ref}_T$ and
$X^{\Delta t}_T$ at the final time can then be computed by summing the
appropriate terms using the reference paths:
\begin{minted}[linenos=true,bgcolor=bg,escapeinside=!!]{matlab}
XTref(1,:) = sum(WH(1:end-1,:) .* diff(W(1:end,:)));
XTdt = nan(numDt,M);
for j=1:numDt
    XTdt(j,:) =  sum(WH(1:2^(j+gap-1):end-1,:) .* diff(W(1:2^(j+gap-1):end,:)));
end
\end{minted}
The code referenced above can be found at the git repository:\\
\url{https://bitbucket.org/datainformeduq/rbwc_code}.

\section{Bound auxiliary functions $f$ by $\varphi$}%
\label{sec:bound-f-varphi}%

\begin{lemma}
  \label{lem:bound-f-by-phi}
  Let $\varphi(x)$ be the payoff function, and define
  \begin{equation*}
    \begin{split}
      f_s(Y^{l_1}_{t_i}, \dots, Y^{l_k}_{t_i}) \defeq
      \E\bigl[\nu(\bar{\bZ}_{t_i}, s)\given Y^{l_1}_{t_i}, \dots,
      Y^{l_k}_{t_i} \bigr] = \E\bigl[
      \E\bigl[\varphi^{(m)}(\widehat{X}_T) \given \bZ_{s} =
      (\bar{X}_{t_i}, \bY_{t_i})\bigr] \given Y^{l_1}_{t_i}, \dots,
      Y^{l_k}_{t_i}\bigr]\,,
    \end{split}
  \end{equation*}
  for components
  $(Y^{l_1}_{t_i}, \dots, Y^{l_k}_{t_i}) \subset \bY_{t_i}$. Then for
  a multiindex $\alpha = (\alpha_1, \dots, \alpha_k)$,
  \begin{equation*}
    |\partial^\alpha f^{kl}_s(\boldsymbol{0})| \lesssim
    \sum_{j = 0}^{|\alpha|}\| \varphi^{(m+j)} \|_\infty \,.
  \end{equation*}
\end{lemma}

\begin{proof}(\cref{lem:bound-f-by-phi}) Recall that
  $\nu(\bar{\bZ}_{t_i}, s)$ is a deterministic function of the jointly
  Gaussian $\eta$-dimensional vector (with
  $\eta \defeq (N_L + 1)\times (i+1)$),
  \begin{equation*}
    \label{eq:lemma-xi}
    \bxi = (\bY_{\tau}, \Delta W_\tau)_{\tau = t_0, \dots, t_i} \,,
  \end{equation*}
  (cf.~\cref{eq:direct-sampling-joint}). Denoting the density as
  $\varrho(\bxi)$, we note $\bxi$ is mean zero and has
  variance-covariance $\bSigma$ given by the known quantities
  $\cov(Y^l_{t_i}, Y^k_{t_j})$, $\cov(Y^l_{t_i}, \Delta W_{t_j})$, and
  $\cov(\Delta W_{t_i}, \Delta W_{t_j})$, which have closed form
  expressions. For each $s$, the variable $\nu(\bar{\bZ}_{t_i}, s)$
  has a density proportional to $\varrho(\bxi)$,
  \begin{equation*}
    \dd{P_{\nu}} \propto \varrho(\bxi) \dd{\bxi}\,.
  \end{equation*}
  Writing $\by = (Y^{l_1}_{t_i}, \dots, Y^{l_k}_{t_i})$, a subset of
  components of $\bY_{t_i}$ of size $k = |\by|$ such that
  $k \leq N_L$, the function
  \begin{equation*}
    f_s(\by) = \E[\nu(\bar{\bZ}_{t_i}, s)\given \by]\,,
  \end{equation*}
  a deterministic function of $\by$, can be expressed in terms of a
  conditional Gaussian density. Partitioning
  $\bxi = (\widetilde{\bxi}, \by)$ and, likewise, the covariance matrix
  $\bSigma$ into components $\bSigma_{11}$ corresponding to
  $\widetilde{\bxi}$, $\bSigma_{22}$ to $\by$, and $\bSigma_{12}$ to the
  mixed terms, the the conditional Gaussian density
  $\varrho(\widetilde{\bxi} \given \by)$ has conditional mean
  \begin{equation*}
    \widetilde{\bmu} = \bSigma_{12}\bSigma_{22}^{-1} \by\,,
  \end{equation*}
  which is linear in $\by$, and conditional variance-covariance matrix
  \begin{equation*}
    \widetilde{\bSigma}
    = \bSigma_{22} - \bSigma_{12}^\top \bSigma_{11}^{-1} \bSigma_{12}\,,
  \end{equation*}
  that does not depend on $\by$.

  Rewritten in terms of this conditional density, $f$ is given by
  \begin{equation*}
    \label{eq:f-wrt-gauss-density}
    f_s(\by) = \int_{\rset^{\eta-k}}
    \nu( \bar{X}_{t_i}, \widetilde{\bxi}, \by; s)
    \varrho(\widetilde{\bxi} \given \by) \dd{\widetilde{\bxi}}\,, 
  \end{equation*}
  where
  \begin{equation*}
    \nu(\bar{X}_{t_i}, \widetilde{\bxi}, \by; s)
    = \E \bigl[ \varphi^{(m)}(\widehat{X}_T) \given \bZ_{s}
    = (\bar{X}_{t_i}, \widetilde{\bxi}, \by) \bigr]\,.
  \end{equation*}
  For a multiindex $\alpha = (\alpha_1, \dots, \alpha_k)$
  corresponding to the components of $\by$,
  \begin{equation*}
    \partial^\alpha = \frac{\partial^\alpha}{\partial \by^{\alpha}}
    = \frac{\partial^{\alpha_1}}{\partial y_{l_1}}
    \cdots \frac{\partial^{\alpha_k}}{\partial y_{l_k}}\,,
  \end{equation*}
  taking the derivative inside the integral we obtain
  \begin{equation}
    \label{eq:lemma-deriv-f}
    \begin{split}
      \partial^\alpha f_s(\by) &= \int_{\rset^{\eta-k}}
      \left\{\bigl(\partial^\alpha \nu(\bar{X}_{t_i}, \widetilde{\bxi},
        \by)\bigr) \varrho(\widetilde{\bxi} \given \by) +
        \nu(\bar{X}_{t_i}, \widetilde{\bxi}, \by) \partial^\alpha
        \varrho(\widetilde{\bxi}
        \given \by) \right\} \dd{\widetilde{\bxi}} \\
      &= \int_{\rset^{\eta-k}} \left\{ \partial^\alpha
        \nu(\bar{X}_{t_i}, \widetilde{\bxi}, \by) + \nu(\bar{X}_{t_i},
        \widetilde{\bxi}, \by) P_k(\by) \right\} \varrho(\widetilde{\bxi}
      \given \by) \dd{\widetilde{\bxi}}\,,
    \end{split}
  \end{equation}
  where
  \begin{equation*}
    \partial^\alpha \varrho(\widetilde{\bxi} \given \by)
    = P_{k} (\by) \varrho(\widetilde{\bxi} \given \by)\,,
  \end{equation*}
  for $P_{k}$ a polynomial of degree $k$ since
  \begin{equation*}
    \varrho(\widetilde{\bxi}\given \by)
    \propto \exp[-\tfrac{1}{2} (\widetilde{\bxi}-\widetilde{\bmu})^\top
    \widetilde{\bSigma}
    (\widetilde{\bxi} - \widetilde{\bmu})]
    = \exp[-\tfrac{1}{2} (\widetilde{\bxi}-\bSigma_{12}\bSigma_{22}^{-1}\by)^\top
    \widetilde{\bSigma}
    (\widetilde{\bxi} - \bSigma_{12}\bSigma_{22}^{-1} \by)]\,.
  \end{equation*}
  The remaining derivative in \cref{eq:lemma-deriv-f} follows
  similarly to the computation of the fluxes in
  \cref{lem:fluxes}, %
  \begin{equation*}
    \partial^{\alpha} \nu(\bar{X}_{t_i}, \widetilde{\bxi}, \by; s)
    = c_{H}^{|\alpha|} \E \Bigl[ \varphi^{(|\alpha|+m)}(\widehat{X}_{T})
    \prod_{j=1}^{|\alpha|}\bigl(\Delta\theta_{l_j} M^{l_j}_{s,T}\bigr)^{\alpha_j} 
    \given \bZ_s =  (\bar{X}_{t_i}, \widetilde{\bxi}, \by) \Bigr]\,,
  \end{equation*}
  and the estimate follows.
\end{proof}

 \end{document}